\documentclass[journal,draftcls,onecolumn,12pt,twoside]{IEEEtran}

\usepackage{times}
\usepackage{amsmath}
\usepackage{amssymb}
\usepackage{amsthm}
\usepackage{color}
\usepackage{algorithm}
\usepackage[noend]{algorithmic}
\usepackage{graphicx}
\usepackage{subfig}
\usepackage{multirow}
\usepackage[bookmarks=false,colorlinks=false,pdfborder={0 0 0}]{hyperref}
\usepackage{cite}
\usepackage{bm}
\usepackage{arydshln}
\usepackage{mathtools}
\usepackage{microtype}

\newtheorem{theorem}{Theorem}

\newtheorem{lemma}[theorem]{Lemma}

\long\def\symbolfootnote[#1]#2{\begingroup
\def\thefootnote{\fnsymbol{footnote}}\footnote[#1]{#2}\endgroup}

\title{Binary MDS Array Codes with Optimal Repair}

\author{Hanxu Hou and Patrick P. C. Lee}

\begin{document}

\maketitle
\vspace{-0.5cm}
\begin{abstract}\symbolfootnote[0]{Hanxu Hou is with the School of Electrical Engineering \& Intelligentization, Dongguan University of Technology and the Department of Computer Science and Engineering, The Chinese University of Hong Kong (E-mail: houhanxu@163.com). Patrick P. C. Lee is with the Department of Computer Science and Engineering, The Chinese University of Hong Kong (E-mail: pclee@cse.cuhk.edu.hk). 
This work was partially supported by the
National Natural Science Foundation of China (No. 61701115) and Research Grants Council of Hong Kong (GRF 14216316 and CRF C7036-15G).}
Consider a binary maximum distance separable (MDS) array code composed of an
$m\times (k+r)$ array of bits with $k$ information columns and $r$ parity
columns, such that any $k$ out of $k+r$ columns suffice to reconstruct the $k$
information columns.  Our goal is to provide {\em optimal repair access} for
binary MDS array codes, meaning that the bandwidth triggered to repair any
single failed information or parity column is minimized.  In this paper, we
propose a generic transformation framework for binary MDS array codes, using
EVENODD codes as a motivating example, to support optimal repair access for
$k+1\le d \le k+r-1$, where $d$ denotes the number of non-failed columns that
are connected for repair; note that when $d<k+r-1$, some of
the chosen $d$ columns in repairing a failed column are specific.
In addition, we show how our transformation framework applies
to an example of binary MDS array codes with asymptotically optimal repair
access of any single information column and enables asymptotically or exactly
optimal repair access for any column.
Furthermore, we present a new transformation for EVENODD codes
with two parity columns such that the existing efficient repair property of
any information column is preserved and the repair access of parity column is
optimal.
\end{abstract}


\IEEEpeerreviewmaketitle

\section{Introduction}
\label{sec:intro}

Large-scale storage systems typically introduce redundancy into data storage
to provide fault tolerance and maintain storage reliability.  Erasure coding
is a redundancy technique that significantly achieves higher reliability than
replication at the same storage overhead \cite{weatherspoon2002}, and has been
widely adopted in commercial storage systems
\cite{huang2012,sathiamoorthy2013}.  One important class of erasure codes is
{\em maximum distance separable (MDS)} codes, which achieve the maximum
reliability for a given amount of redundancy.  Specifically, an MDS code
transforms $k$ information symbols into $k+r$ encoded symbols of the same size
for some configurable parameters $k$ and $r$, such that any $k$ out of
$k+r$ symbols are sufficient to retrieve all $k$ information symbols.
Reed-Solomon (RS) codes \cite{reed1960} are one well-known example of MDS codes.

In this paper, we examine a special class of MDS codes called {\em binary MDS
array codes}, which have low computational complexity since the encoding and
decoding procedures only involve XOR operations.  Examples of binary MDS array
codes are EVENODD \cite{blaum1995,blaum1996,blaum2001}, X-code \cite{xu1999},
and RDP \cite{corbett2004,blaum2006}.  Specifically, we consider a binary MDS
array code that is composed of an array of size $m\times (k+r)$, where each
element in the array is a bit.  In this work, we assume that the code is
{\em systematic}, meaning that $k$ columns are {\em information columns} that
store information bits, and the remaining $r$ columns are {\em parity columns}
that store parity bits encoded from the $k$ information columns.  The code is
MDS, meaning that any $k$ out of $k+r$ columns can reconstruct all the
original $k$ information columns.  We distribute the $k+r$ columns across
$k+r$ distinct storage nodes, such that the bits in each column are stored in
the same node. We use the terms ``column'' and ``node'' interchangeably in
this paper.

In large-scale storage systems, node failures are common and the majority of
all failures are single node failures \cite{ford2010}. Thus, it is critical to
design an efficient repair scheme for repairing the lost bits of a single failed
node, while providing fault tolerance for multiple node failures.  The problem
of repairing a single node failure was first formulated by Dimakis \emph{et
al.} \cite{dimakis2010}, in which it is shown that the amount of symbols
downloaded for repairing a single node failure (called the \emph{repair
bandwidth}) of an $m\times (k+r)$ MDS array code is at least (in units of
bits):
\begin{equation}
\frac{dm}{d-k+1},
\label{eq:rep}
\end{equation}
where $d$ ($k\leq d\leq k+r-1$) is the number of nodes connected to repairing
the failed node.  Many constructions
\cite{rashmi2011,suh2011,hou2016a,li2017,ye2017} of MDS array codes
have been proposed to achieve the optimal repair bandwidth in \eqref{eq:rep}.
If the repair bandwidth of a binary MDS array code achieves the optimal value
in \eqref{eq:rep}, we say that the code has optimal repair bandwidth. If the
repair does not require any arithmetic operations on the $d$ connected nodes,
then the repair is called \emph{uncoded repair}. A binary MDS array code is
said to achieve \emph{optimal repair access} if the repair bandwidth is
\eqref{eq:rep} with uncoded repair.


\subsection{Related Work}

There are many related studies on binary MDS array codes along different
directions, such as new constructions
\cite{blaum1996,blaum2001,atul2012,hou2014,hou2018}, efficient decoding
methods \cite{huang2008,Jiang2013,wang2012,huang2016,hou2018d,hou2018a} and
the improvement of the repair problem
\cite{wang2010,xiang2010,xiang2011,zhu2014,hou2017,hou2018b,hou2018c}.

In particular, EVENODD is well explored in the literature, and has been
extended to STAR codes \cite{huang2008} with three parity columns and
generalized EVENODD \cite{blaum1996,blaum2001} with more parity columns.  The
computational complexity of EVENODD is optimized in \cite{hou2018} by a new
construction. A sufficient condition for the generalized EVENODD to be
MDS with more than eight parity columns is given in \cite{hou2016}.

RDP is another important class of binary MDS array codes with two parity
columns. It is extended to RTP codes \cite{atul2012} to tolerate three column
failures.  Blaum \cite{blaum2006} generalized RDP to correct more than
three column failures and showed that the generalized EVENODD and
generalized RDP share the same MDS property condition. The authors in
\cite{hou2018d} proposed a unified form of generalized EVENODD and
generalized RDP, and presented an efficient decoding method for some
patterns of failures.

The above constructions are based on the Vandermonde matrix. Some
constructions of binary MDS array codes based on Cauchy matrix are Cauchy
Reed-Solomon codes \cite{blomer1999}, Rabin-like codes
\cite{feng2005,hou2018a} and circulant Cauchy codes \cite{schindelhauer2013}.

Most of the decoding methods focus on generalized EVENODD
\cite{Jiang2013,huang2008} and generalized RDP \cite{atul2012,huang2016} with
three parity columns. The study \cite{hou2018d} shows an efficient
erasure decoding method based on the LU factorization of
Vandermonde matrix for EVENODD and RDP with more than two parity columns.

There have been many studies
\cite{xiang2010,xu2014,wang2010,wang2013,wang2016,gad2013,pamies2016,hou2017,hou2018b,hou2018c}
on the repair problem of binary MDS array codes.  Some optimal repair schemes
reduce I/O for RDP \cite{xiang2010}, X-code \cite{xu2014} and EVENODD
\cite{wang2010} by approximately 25\%, but the repair bandwidth is
sub-optimal. ButterFly codes \cite{gad2013,pamies2016} are binary MDS array
codes with optimal repair for information column failures, but only has two
parity columns (i.e., $r\!=\!2$). MDR codes \cite{wang2013,wang2016} are
constructed with $r=2$ and have optimal repair bandwidth for $k$ information
columns and one parity column.  Binary MDS array codes with more than two
parity columns are proposed in \cite{hou2017,hou2018c,hou2018b}; however, the
repair bandwidth is asymptotically optimal and the $d$ helper columns are
specifically selected.

\subsection{Contributions}

The contributions of this paper are summarized as follows.
\begin{enumerate}
\item
First, we propose a generic transformation for an $m\times (k+r)$ EVENODD
code.  The transformed EVENODD code is of size $m(d-k+1)\times (k+r)$ and has
three properties: (1) the transformed EVENODD code achieves optimal repair
access for the chosen $d-k+1$ columns; (2) the property of
optimal repair access for the chosen $d-k+1$ columns of the
transformed EVENODD code is preserved if we apply the transformation once more
for the transformed EVENODD code; and (3) the transformed EVENODD code is MDS.
\item
Second, we present a family of $m(d-k+1)^{\lceil \frac{k}{d-k+1}\rceil+\lceil
\frac{r}{d-k+1}\rceil}\times (k+r)$ multi-layer transformed EVENODD codes with
$r\geq 2$, such that it achieves optimal repair access for all columns based
on the EVENODD transformation, where $k+1\leq d\leq k+r-1$. Some of the $d$
helper columns need to be specifically selected.
\item
Third, the efficient decoding method of the original EVENODD code is also
applicable to the proposed family of multi-layer transformed EVENODD codes.
\item
Lastly, the other binary MDS array codes, such as RDP \cite{blaum2006} and
codes in
\cite{blomer1999,feng2005,hou2018a,schindelhauer2013,hou2017,hou2018c,hou2018b},
can also be transformed to achieve optimal repair access and the efficient
decoding methods of the original binary MDS array codes are maintained in the
transformed codes.
By applying the transformation with well-chosen encoding
coefficients for an example of binary MDS array codes \cite{hou2018c} that have
asymptotically optimal repair access for any information column, we show that
the obtained transformed codes have asymptotically optimal repair access for
any information column and optimal repair access for any parity column.  We
also show how to design a transformation for EVENODD codes with two parity
columns such that the transformed codes have optimal repair access for any
single parity column and the repair access of any single information column of
the transformed codes is roughly $3/4$ of all the information bits.
\end{enumerate}

A closely related work to ours is \cite{li2017}, which also
proposes a transformation for non-binary MDS codes to enable optimal repair
access.  The main differences between the work in \cite{li2017} and ours are
two-fold.  First, our transformation is designed for binary MDS array codes,
while the transformation in \cite{li2017} is designed for non-binary MDS
codes. The minimum operation unit of our transformed codes is a bit, so that
we can carefully choose the encoding coefficients of the transformation to
combine the efficient repair property of existing or newly designed binary MDS
array codes for any single information column as well as the optimal repair of
the transformed codes for any parity column. In contrast, the minimum
operation unit of the transformation \cite{li2017} for non-binary MDS codes is
a field element, so we cannot directly apply the transformation \cite{li2017}
for binary MDS array codes. Even though we can view each column of some binary
MDS array codes (such as EVENODD codes with $p$ being a special prime number
\cite{blaum1996}) as a field element, if we apply the transformation
\cite{li2017} for such binary MDS array codes, the efficient repair property
of such binary MDS array codes for any single information column cannot be
maintained, as the efficient repair property of binary MDS array codes is
achieved by downloading some bits from the chosen columns but not all the
bits (field element) from the chosen columns. We illustrate the transformation
of an example of binary MDS array codes \cite{hou2018c} with asymptotically
optimal repair access for any single information column to obtain the
transformed array codes that have asymptotically optimal repair for any
information column and optimal repair access for any parity column in
Section~\ref{sec:trans-exm}. We also design a new transformation for EVENODD
codes with $r=2$ parity columns such that the repair access of any single
information column of the transformed codes is roughly $3/4$ of all the
information bits and the repair access of each parity column is optimal in
Section \ref{sec:trans-evenodd}. Second, our work allows a more flexible
number of nodes connected for repairing the failed node. In particular, our
work allows $k+1\leq d\leq k+r-1$, while the work in \cite{li2017} requires
$d=k+r-1$.

\section{Transformation of EVENODD Codes}
\label{sec:framework}

We first review the definition of EVENODD codes. We then present our
transformation approach.

\subsection{Review of EVENODD Codes}

An EVENODD code is an array code of size $(p-1)\times (k+r)$, where $p$ is a
prime number with $p\geq \max\{k,r\}$. Given the $(p-1)\times (k+r)$ array
$[a_{i,j}]$ for $i=0,1,\ldots,p-2$ and $j=0,1,\ldots,k+r-1$, the $p-1$ bits
$a_{0,j},a_{1,j},\ldots,a_{p-2,j}$ in column $j$ can be represented as a
polynomial
\[
a_{j}(x)=a_{0,j}+a_{1,j}x+\cdots+a_{p-2,j}x^{p-2}.
\]
Without loss of generality, we store the information bits in the $k$ leftmost
columns and the parity bits in the remaining $r$ columns.  The first $k$
polynomials $a_{0}(x),\ldots,a_{k-1}(x)$ are called \emph{information
polynomials}, and the last $r$ polynomials $a_{k},\ldots,a_{k+r-1}(x)$ are
\emph{parity polynomials}. The $r$ parity polynomials are computed as
\begin{equation}
\begin{array}{ll}
\begin{bmatrix}
a_{k}(x)& \cdots & a_{k+r-1}(x)\end{bmatrix}
=\begin{bmatrix}a_{0}(x)& \cdots & a_{k-1}(x)\end{bmatrix}
\begin{bmatrix}
1& 1& \cdots &1\\
1& x& \cdots &x^{r-1}\\
\vdots & \vdots & \ddots & \vdots \\
1& x^{k-1} & \cdots & x^{(r-1)(k-1)}
\end{bmatrix}
\end{array}
\label{eq:evenodd-en}
\end{equation}
over the ring $\mathbb{F}_2[x]/(1+x+\cdots+x^{p-1})$. The
matrix on the right-hand side of \eqref{eq:evenodd-en} is called the
\emph{encoding matrix}.

\subsection{The Transformation}
\label{sec:trans}
We will present the transformation that can convert a $(p-1)\times (k+r)$
EVENODD code into a $(p-1)(d-k+1)\times (k+r)$ transformed code with
optimal repair access for any chosen $d-k+1$ columns, where $k+1\leq d\leq
k+r-1$. For the ease of presentation, we assume that the chosen $d-k+1$ columns
are the first $d-k+1$ columns in the following discussion.

\subsubsection{The First Transformation}
\label{sec:first-trans}
Given the codewords of a $(p-1)\times (k+r)$ EVENODD code
$a_{0}(x),\ldots,a_{k+r-1}(x)$, we first generate $d-k+1$ instances
$a_{0,\ell}(x),\ldots,a_{k+r-1,\ell}(x)$ for $\ell=0,1,\ldots,d-k$.
Specifically, the $r$ parity polynomials $a_{k,\ell}(x),\ldots,a_{k+r-1,\ell}(x)$ are
computed by the multiplication of $[a_{0,\ell}(x),\ldots,a_{k-1,\ell}(x)]$ and the
encoding matrix in \eqref{eq:evenodd-en}, where $\ell=0,1,\ldots,d-k$. For
$i=0,1,\ldots,d-k$, the polynomials in column $i$ are
\begin{equation}
\begin{array}{ll}
a_{i,0}(x)+a_{0,i}(x),a_{i,1}(x)+a_{1,i}(x),\ldots,a_{i,i-1}(x)+a_{i-1,i}(x),& \\
a_{i,i}(x),a_{i,i+1}(x)+(1+x^e)a_{i+1,i}(x),& \\
a_{i,i+2}(x)+(1+x^e)a_{i+2,i}(x),\ldots,a_{i,d-k}(x)+(1+x^e)a_{d-k,i}(x),&
\end{array}
\label{eq:trans1}
\end{equation}
where $e$ is a positive integer with $1\leq e\leq p-1$.
On the other hand, for $i=d-k+1,\ldots,k+r-1$, the polynomials in column $i$
are
\[
a_{i,0}(x),a_{i,1}(x),\ldots,a_{i,d-k}(x).
\]
The above transformation is called \emph{the first transformation} and the obtained codes are called transformed EVENODD codes. Each column of the transformed EVENODD codes has $d-k+1$ polynomials. Table
\ref{table:A1} shows an example of the first transformed EVENODD code with
$k=4$, $r=2$, $d=5$ and $e=1$.

\begin{table*}[!t]
\caption{The first transformation for EVENODD codes with
$k=4$, $r=2$, $d=5$ and $e=1$.}
\vspace{-8pt}
\begin{center}
\begin{tabular}{|c|c|c|c|c|c|}
\hline
Column 0 & Column 1  & Column 2  & Column 3 & Column 4  & Column 5 \\
\hline
$a_{0,0}(x)$& $a_{1,0}(x)+$ & $a_{2,0}(x)$& $a_{3,0}(x)$ & $a_{4,0}(x)=a_{0,0}(x)+a_{1,0}(x)+$ & $a_{5,0}(x)=a_{0,0}(x)+xa_{1,0}(x)+$ \\
 &$a_{0,1}(x)$ & & & $a_{2,0}(x)+a_{3,0}(x)$& $x^2a_{2,0}(x)+x^3a_{3,0}(x)$ \\
\hline
$a_{0,1}(x)+$& $a_{1,1}(x)$ &$a_{2,1}(x)$& $a_{3,1}(x)$ & $a_{4,1}(x)=a_{0,1}(x)+a_{1,1}(x)+$ & $a_{5,1}(x)=a_{0,1}(x)+xa_{1,1}(x)+$\\
$(1+x)a_{1,0}(x)$ & & & & $a_{2,1}(x)+a_{3,1}(x)$& $x^2a_{2,1}(x)+x^3a_{3,1}(x)$ \\
\hline
\end{tabular}
\end{center}
\label{table:A1}
\end{table*}

\textbf{Remark}. For $i< j\in \{0,1,\ldots,d-k\}$, columns $i$ and $j$ contain the following two polynomials
\begin{align*}
a_{i,j}(x)+(1+x^e)a_{j,i}(x), a_{j,i}(x)+a_{i,j}(x).
\end{align*}
We can solve $x^ea_{j,i}(x)$ by summing the above two polynomials. Then, we can
obtain $a_{j,i}(x)$ by multiplying $x^ea_{j,i}(x)$ by $x^{p-e}$, and
$a_{i,j}(x)$ by summing $a_{j,i}(x)+a_{i,j}(x)$ and $a_{j,i}(x)$. Therefore,
we can solve two information polynomials $a_{j,i}(x),a_{i,j}(x)$ from columns
$i$ and $j$. If $\ell$
columns are chosen that are in the first $d-k+1$ columns, then we can solve
$\ell(\ell-1)$ information polynomials from the chosen $\ell$ columns, where
$\ell=2,3,\ldots,d-k+1$.

\subsubsection{The Second Transformation}

Note that the above transformed code is a non-systematic code. To obtain the
systematic code, we can first replace $a_{i,\ell}(x)+a_{\ell,i}(x)$ by
$a'_{i,\ell}(x)$ and replace $a_{\ell,i}(x)+(1+x^e)a_{i,\ell}(x)$ by $a'_{\ell,i}(x)$
for $\ell<i$, to obtain that
\begin{equation}
\left\{
\begin{array}{ll}
a_{i,\ell}(x)=x^{p-e}a'_{i,\ell}(x)+x^{p-e}a'_{\ell,i}(x), \\
a_{\ell,i}(x)=(1+x^{p-e})a'_{i,\ell}(x)+x^{p-e}a'_{\ell,i}(x).
\end{array}
\right.
\label{eq:substi}
\end{equation}
Then, we can show the equivalent systematic transformed code as follows. For $i=0,1,\ldots,k-1$, the $d-k+1$ polynomials in column $i$ are $a_{i,\ell}(x)$ for $\ell=0,1,\ldots,d-k$. Recall that the polynomial $a_{i,\ell}(x)$ is computed by
\[
a_{i,\ell}(x)=\sum_{j=0}^{k-1}x^{j(i-k)}a_{j,\ell}(x)
\]
for $i=k,k+1,\ldots,k+r-1$ and $\ell=0,1,\ldots,d-k$. We update $r(d-k+1)$ polynomials $a_{i,\ell}(x)$ for $i=k,k+1,\ldots,k+r-1$ and $\ell=0,1,\ldots,d-k$, by replacing the component $x^{j(i-k)}a_{j,\ell}(x)$ of $a_{i,\ell}(x)$ by $x^{j(i-k)}(x^{p-e}a_{j,\ell}(x)+(1+x^{p-e})a_{\ell,j}(x))$ for $j<\ell$, and replacing the component $x^{j(i-k)}a_{j,\ell}(x)$ of $a_{i,\ell}(x)$ by $x^{j(i-k)}(x^{p-e}a_{j,\ell}(x)+x^{p-e}a_{\ell,j}(x))$ for $j> \ell$, i.e.,
\begin{align*}
&a_{i,\ell}(x)=\left(\sum_{j=0}^{\ell-1}x^{j(i-k)}(x^{p-e}a_{j,\ell}(x)+(1+x^{p-e})a_{\ell,j}(x)) \right)+\\
&x^{\ell(i-k)}a_{\ell,\ell}(x)+\left(\sum_{j=\ell+1}^{d-k}x^{j(i-k)}(x^{p-e}a_{j,\ell}(x)+x^{p-e}a_{\ell,j}(x))\right)\\
&+\left(\sum_{j=d-k+1}^{k-1}x^{j(i-k)}a_{j,\ell}(x)\right).
\end{align*}
The $d-k+1$ polynomials in column $i$ for $i=k,k+1,\ldots,k+r-1$ are
$a_{i,\ell}(x)$ for $\ell=0,1,\ldots,d-k$. In other words, the $d-k+1$
polynomials in column $i$ for $i=k,k+1,\ldots,k+r-1$ are
$a_{i,0}(x),\ldots,a_{i,d-k}(x)$, which are computed by \eqref{eq:trans} over the ring $\mathbb{F}_2[x]/(1+x+\cdots+x^{p-1})$.
\newcounter{mytempeqncnt}
\begin{figure*}[!t]
\begin{equation}
\small
\begin{array}{ll}
\begin{bmatrix}
1 & x^{i-k} & \cdots & x^{(k-1)(i-k)}
\end{bmatrix} \cdot \\
\begin{bmatrix}
a_{0,0}(x) & x^{p-e}a_{0,1}(x)+(1+x^{p-e})a_{1,0}(x) & \cdots & x^{p-e}a_{0,d-k}(x)+(1+x^{p-e})a_{d-k,0}(x) \\
x^{p-e}(a_{1,0}(x)+a_{0,1}(x)) & a_{1,1}(x) & \cdots & x^{p-e}a_{1,d-k}(x)+(1+x^{p-e})a_{d-k,1}(x)\\
\vdots & \vdots & \ddots & \vdots \\
x^{p-e}(a_{d-k,0}(x)+a_{0,d-k}(x)) & x^{p-e}(a_{d-k,1}(x)+a_{1,d-k}(x)) & \cdots & a_{d-k,d-k}(x) \\
a_{d-k+1,0}(x) & a_{d-k+1,1}(x) & \cdots & a_{d-k+1,d-k}(x)\\
\vdots & \vdots & \ddots & \vdots \\
a_{k-1,0}(x) & a_{k-1,1}(x) & \cdots & a_{k-1,d-k}(x)\\
\end{bmatrix}.
\end{array}
\label{eq:trans}
\end{equation}
\end{figure*}
The above transformation is called \emph{the second transformation}. The transformed EVENODD code is denoted by $\mathsf{EVENODD}_1$.
Note that each column of EVENODD codes has one polynomial, and each column of $\mathsf{EVENODD}_1$ obtained by applying the transformation for EVENODD codes has $d-k+1$ polynomials.

When $k=4$, $r=2$, $d=5$ and $e=1$, the $\mathsf{EVENODD}_1$ code with the second transformation is shown in Table \ref{table:A}.
We claim that we can recover all the information polynomials from any four columns. We can obtain the information polynomials from columns 0, 1, 2 and 3 directly. Consider that we want to recover the information polynomials from one parity column and three information columns, say columns 0, 2, 3 and 4. We can obtain $a_{1,0}(x)$ by
\[
x(a_{4,0}(x)+a_{0,0}(x)+x^{p-1}a_{0,1}(x)+a_{2,0}(x)+a_{3,0}(x)),
\]
and $a_{1,1}(x)$ by
\[
a_{4,1}(x)+x^{p-1}a_{0,1}(x)+(1+x^{p-1})a_{1,0}(x)+a_{2,1}(x)+a_{3,1}(x).
\]
Suppose that we want to solve the information polynomials from two information columns and two parity columns, say columns 1, 2, 4 and 5. First, we compute the following two polynomials by subtracting $a_{1,1}(x),a_{2,1}(x)$ from $a_{4,1}(x),a_{5,1}(x)$,
\begin{align*}
p_1(x)=&a_{4,1}(x)+(1+x^{p-1})a_{1,0}(x)+a_{1,1}(x)+a_{2,1}(x)\\
=&x^{p-1}a_{0,1}(x)+a_{3,1}(x),\\
p_2(x)=&a_{5,1}(x)+(1+x^{p-1})a_{1,0}(x)+xa_{1,1}(x)+x^2a_{2,1}(x)\\
=&x^{p-1}a_{0,1}(x)+x^3a_{3,1}(x).
\end{align*}
Then, we can solve $a_{3,1}(x)$ by $\frac{p_1(x)+p_2(x)}{1+x^3},$\footnote{$1+x^3$ is invertible in $\mathbb{F}_2[x]/(1+x+\cdots+x^{p-1})$ due to the MDS property of EVENODD codes given in Proposition 2.2 in \cite{blaum1996}.} and $a_{0,1}(x)$ by $x(a_{3,1}(x)+p_1(x))$. The other two information polynomials $a_{0,0}(x),a_{3,0}(x)$ can be solved similarly.

The repair access of each of the first two columns is optimal. Suppose that the first column fails. We can first solve $a_{0,0}(x)$ and $x^{p-1}(a_{1,0}(x)+a_{0,1}(x))$ by accessing four polynomials $a_{2,0}(x),a_{3,0}(x),a_{4,0}(x),a_{5,0}(x)$ due to the MDS property of EVENODD codes, and then recover $a_{0,1}(x)$ by computing $x(x^{p-1}(a_{1,0}(x)+a_{0,1}(x))+x^{p-1}a_{1,0}(x))$. Therefore, we can recover two information polynomials by downloading five polynomials from five helper columns, and the repair bandwidth achieves the minimum value in \eqref{eq:rep}. The repair of the second column is similar.

\subsection{Properties of Transformed EVENODD Codes}
\label{sec:prop}
The next theorem shows that the second transformed EVENODD code is also MDS code.


\begin{table*}[t]
\scriptsize
\caption{The second transformation for EVENODD codes with $k=4$, $r=2$, $d=5$ and $e=1$.}
\vspace{-8pt}
\begin{center}
\begin{tabular}{|c|c|c|c|c|c|}
\hline
Column 0 & Column 1  & Column 2  & Column 3 & Column 4  & Column 5 \\
\hline
$a_{0,0}(x)$& $a_{1,0}(x)$ &$a_{2,0}(x)$& $a_{3,0}(x)$ & $a_{0,0}(x)+x^{p-1}a_{1,0}(x)+$ & $a_{0,0}(x)+a_{1,0}(x)+a_{0,1}(x)$ \\
& & & & $x^{p-1}a_{0,1}(x)+a_{2,0}(x)+a_{3,0}(x)$& $+x^2a_{2,0}(x)+x^3a_{3,0}(x)$\\
\hline
$a_{0,1}(x)$& $a_{1,1}(x)$ &$a_{2,1}(x)$& $a_{3,1}(x)$ & $x^{p-1}a_{0,1}(x)+(1+x^{p-1})a_{1,0}(x)+$ & $x^{p-1}a_{0,1}(x)+(1+x^{p-1})a_{1,0}(x)+$\\
& & & & $a_{1,1}(x)+a_{2,1}(x)+a_{3,1}(x)$& $xa_{1,1}(x)+x^2a_{2,1}(x)+x^3a_{3,1}(x)$\\
\hline
\end{tabular}
\end{center}
\label{table:A}
\end{table*}

\begin{theorem}
If the $(k+r,k)$ EVENODD code is MDS code, then the second transformed EVENODD
code is also MDS.
\label{thm:mds}
\end{theorem}
\begin{proof}
The code is an MDS code if any $k$ out of $k+r$ columns can retrieve all
information bits. It is equivalent to show that the $k$ information columns
can be reconstructed from any $t$ information columns and any $k-t$ parity
columns, where $\max\{0,k-r\}\leq t\leq k$.
When $t=k$, we can obtain the $k$ information columns directly.

In the following, we consider the case of $t< k$. Suppose that columns $i_1,i_2,\ldots,i_{t}$ and columns $j_1,j_2,\ldots,j_{k-t}$ are connected with $0\leq i_1<\ldots < i_{t}\leq k-1$ and $k\leq j_1<\ldots<j_{k-t}\leq k+r-1$. We need to recover $k-t$ information columns $e_1,e_2,\ldots,e_{k-t}$, where $$e_1<e_2<\cdots < e_{k-t}\in \{0,1,\ldots,k-1\}\setminus \{i_1,i_2,\ldots,i_{t}\}.$$
Recall that we can obtain $t(d-k+1)$ information polynomials $a_{i_1,\ell}(x),\ldots,a_{i_{t},\ell}(x)$ and $(k-t)(d-k+1)$ parity polynomials $a_{j_1,\ell}(x),\ldots,a_{j_{k-t},\ell}(x)$ from the connected columns, where $\ell=0,1,\ldots,d-k$.

We divide the proof into two cases: $i_1<d-k$ and $i_1\geq d-k$. We first assume that $i_1<d-k$. By subtracting $t(d-k+1)$ information polynomials from $k-t$ parity polynomials $a_{j_{1},i_1}(x),\ldots,a_{j_{k-t},i_1}(x)$ each, we obtain $k-t$ syndrome polynomials over $\mathbb{F}_2[x]/(1+x+\cdots+x^{p-1})$ as
\begin{align*}
\scriptsize
&[
x^{p-e}a_{e_1,i_1}(x) \text{ } \cdots \text{ } x^{p-e}a_{e_{\alpha},i_1}(x) \text{ } a_{e_{\alpha+1},i_1}(x) \text{ } \cdots \text{ } a_{e_{k-t},i_1}(x)
] \\
&\begin{bmatrix}
x^{e_1(j_1-k)} & x^{e_1(j_2-k)} & \cdots & x^{e_1(j_{k-t}-k)}\\
x^{e_2(j_1-k)} & x^{e_2(j_2-k)} & \cdots & x^{e_2(j_{k-t}-k)}\\
\vdots & \vdots & \ddots & \vdots \\
x^{e_{k-t}(j_1-k)} & x^{e_{k-t}(j_2-k)} & \cdots & x^{e_{k-t}(j_{k-t}-k)}
\end{bmatrix},
\end{align*}
where $\alpha$ is an integer that ranges from 1 to $k-t-1$ with $e_{\alpha}\leq d-k$ and $e_{\alpha+1}\geq d-k+1$. In the decoding process from columns 1, 2, 4 and 5 of the example in Table \ref{table:A}, we have $k=4$, $t=2$, $e=1$, $i_1=1$, $e_1=0$, $e_2=3$, $j_1=k$, $j_2=k+1$ and $\alpha=1$. The two syndrome polynomials are
\begin{align*}
\begin{bmatrix}
p_1(x) & p_2(x)
\end{bmatrix}=\begin{bmatrix}
x^{p-1}a_{0,1}(x) & a_{3,1}(x)
\end{bmatrix}\cdot \begin{bmatrix}
1 & 1\\
1 & x^3 \\
\end{bmatrix}.
\end{align*}

As the $(k+r,k)$ EVENODD code is MDS, we can recover the polynomials $$x^{p-e}a_{e_1,i_1}(x), \ldots, x^{p-e}a_{e_{\alpha},i_1}(x), a_{e_{\alpha+1},i_1}(x), \ldots, a_{e_{k-t},i_1}(x),$$
and therefore, $a_{e_1,i_1}(x),a_{e_{2},i_1}(x), \ldots, a_{e_{k-t},i_1}(x)$ can be recovered. Let $c$ be an integer with $2\leq c\leq t$ such that $i_{c-1}<d-k$ and $i_{c}\geq d-k$. By the same argument, we can recover polynomials $a_{e_1,i_{h}}(x),a_{e_{2},i_{h}}(x), \ldots, a_{e_{k-t},i_{h}}(x)$ for $h=2,3,\ldots,c-1$.
Once the polynomials $$a_{e_1,i_{h}}(x),a_{e_{2},i_{h}}(x), \ldots,
a_{e_{k-t},i_{h}}(x)$$ for $h=1,2,\ldots,c-1$ are known, we can recover all
the other failed polynomials by first subtracting all $t(d-k+1)$ information
polynomials and the known polynomials $a_{e_1,i_{h}}(x),a_{e_{2},i_{h}}(x),
\ldots, a_{e_{k-t},i_{h}}(x)$ with $h=1,2,\ldots,c-1$ from parity polynomials
$a_{j_{1},i_{\ell}}(x),\ldots,a_{j_{k-t},i_{\ell}}(x)$, followed by solving
the failed polynomials according to the MDS property of the $(k+r,k)$ EVENODD
code, where $\ell\in \{0,1,\ldots,d-k\}\setminus\{i_1,i_2,\ldots,i_{c-1}\}$.

If $i_1\geq d-k$, then $i_t>\cdots > i_1\geq d-k$ and we can obtain the following syndrome polynomials by subtracting all $t(d-k+1)$ information polynomials from all $(k-t)(d-k+1)$ parity polynomials
\begin{align*}
&\begin{bmatrix}
a^{*}_{e_1,\ell}(x) & a^{*}_{e_2,\ell}(x) & \cdots & a^{*}_{e_{k-t},\ell}(x)
\end{bmatrix}\cdot \\
&\begin{bmatrix}
x^{e_1(j_1-k)} & x^{e_1(j_2-k)} & \cdots & x^{e_1(j_{k-t}-k)}\\
x^{e_2(j_1-k)} & x^{e_2(j_2-k)} & \cdots & x^{e_2(j_{k-t}-k)}\\
\vdots & \vdots & \ddots & \vdots \\
x^{e_{k-t}(j_1-k)} & x^{e_{k-t}(j_2-k)} & \cdots & x^{e_{k-t}(j_{k-t}-k)}
\end{bmatrix},
\end{align*}
where
\begin{equation}
\small
a^{*}_{e_i,\ell}(x) = \left\{
\begin{array}{lll}
a_{e_i,\ell}(x) & \textrm{if } e_i=\ell, \\
x^{p-e}a_{e_i,\ell}(x)+(1+x^{p-e})a_{\ell,e_i}(x) & \textrm{if } e_i<\ell, \\
x^{p-e}a_{e_i,\ell}(x)+x^{p-e}a_{\ell,e_i}(x) & \textrm{if } e_i>\ell,
\end{array}
\right.
\label{eq:parity1}
\end{equation}
$\ell=0,1,\ldots,d-k$.
The polynomials $a^{*}_{e_1,\ell}(x)$, $a^{*}_{e_2,\ell}(x), \ldots,
a^{*}_{e_{k-t},\ell}(x)$ can be recovered, because $(k+r,k)$ EVENODD code
is MDS. Then, we can obtain $a_{\ell,\ell}(x)$ directly, $a_{e_i,\ell}(x)$ for $e_i>\ell$ by $a^{*}_{e_i,\ell}(x)+a^{*}_{\ell,e_i}(x)$, and $a_{e_i,\ell}(x)$ for $e_i<\ell$ by $x^e(a_{\ell,e_i}(x)+a^{*}_{\ell,e_i}(x))$.
\end{proof}

We show in the next theorem that the second transformed EVENODD code has optimal access for the first $d-k+1$ columns.
\begin{theorem}
The repair bandwidth and repair access of column $i$ of the second transformed EVENODD code for $i=0,1,\ldots,d-k$ is optimal.
\label{thm:rep}
\end{theorem}
\begin{proof}
For $i=0,1,\ldots,d-k$, column $i$ can be repaired by  downloading one polynomial from each of $d$ helper columns. Among $d$ columns, $d-k$ columns are columns $0,1,\ldots,i-1,i+1,\ldots,d-k$ and other $k$ columns are chosen from columns $d-k+1,\ldots,k+r-1$. Specifically, we can recover the polynomials $a_{i,i}(x)$, $x^{p-e}a_{j,i}(x)+(1+x^{p-e})a_{i,j}(x)$ for $j<i$ and $x^{p-e}(a_{\ell,i}(x)+a_{i,\ell}(x))$ for $\ell>i$, by downloading $k$ polynomials $a_{h_1,i}(x),\ldots,a_{h_k,i}(x)$ from columns $h_1,\ldots,h_k$, where $h_1\neq \ldots \neq h_k\in \{d-k+1,\ldots,k+r-1\}$, due to the MDS property of EVENODD codes. Then we download $d-k$ polynomials $a_{0,i}(x),\ldots,a_{i-1,i}(x),a_{i+1,i}(x),\ldots,a_{d-k,i}(x)$ from columns $0,1,\ldots,i-1,i+1,\ldots,d-k$. Finally, we subtract the downloaded polynomials $a_{0,i}(x),\ldots,a_{i-1,i}(x),a_{i+1,i}(x),\ldots,a_{d-k,i}(x)$ from the recovered polynomials $x^{p-e}a_{j,i}(x)+(1+x^{p-e})a_{i,j}(x)$ for $j<i$ and $x^{p-e}(a_{\ell,i}(x)+a_{i,\ell}(x))$ for $\ell>i$, to obtain polynomials $a_{i,0}(x),a_{i,1}(x),\ldots,a_{i,i-1}(x),a_{i,i+1}(x),\ldots,a_{i,d-k}(x)$. This completes the proof.
\end{proof}

Note that $\mathsf{EVENODD}_1$ with the first
transformation also satisfies Theorem \ref{thm:mds} and Theorem \ref{thm:rep}, as the two transformations are equivalent. In the following, let $\mathsf{EVENODD}_1$ be the transformed code with the first transformation and
$\mathsf{EVENODD}_2$ be the transformed code by applying the first transformation for the columns from $d-k+1$ to $2d-2k+1$ of $\mathsf{EVENODD}_1$. 
Specifically, we can obtain $\mathsf{EVENODD}_2$ as follows. Let $$t=d-k+1.$$
We first generate $t$ instances of the code $\mathsf{EVENODD}_1$ and view the $t$ polynomials stored in each column of $\mathsf{EVENODD}_1$ as a vector. For $\ell=0,1,\ldots,d-k$ and $h=0,1,\ldots,n-1$, the vector stored in column $h$ of instance $\ell$ of $\mathsf{EVENODD}_1$ is denoted as $\mathbf{v}^\ell_h$. For $i=0,1,\ldots,d-k$, column $t+i$ of $\mathsf{EVENODD}_2$ stores the following $t$ vectors ($t^2$ polynomials)
\begin{equation}
\begin{array}{ll}
& \mathbf{v}^0_{t+i}+\mathbf{v}^i_{t},\\
& \mathbf{v}^1_{t+i}+\mathbf{v}^i_{t+1},\ldots,\\
& \mathbf{v}^{i-1}_{t+i}+\mathbf{v}^i_{t+i-1},\\
& \mathbf{v}^i_{t+i},\\
& \mathbf{v}^{i+1}_{t+i}+(1+x^e)\mathbf{v}^i_{t+i+1},\\
& \mathbf{v}^{i+2}_{t+i}+(1+x^e)\mathbf{v}^i_{t+i+2},\ldots,\\
& \mathbf{v}^{d-k}_{t+i}+(1+x^e)\mathbf{v}^i_{t+d-k},
\end{array}
\label{eq:trans-vectors}
\end{equation}
where $1\leq e\leq p-1$. Note that the multiplication of a polynomial $x^e$ and a vector
\[
\mathbf{v}=\begin{bmatrix}
v_0 & v_1 &\ldots &v_{d-k}
\end{bmatrix}
\]
is defined as
\[
x^e\mathbf{v}=\begin{bmatrix}
x^ev_0 & x^ev_1 &\ldots &x^ev_{d-k}
\end{bmatrix}
\]
and the addition of two vectors
\[
\mathbf{v}^1=\begin{bmatrix}
v^1_0 & v^1_1 &\ldots &v^1_{d-k}
\end{bmatrix}
\]
and
\[
\mathbf{v}^2=\begin{bmatrix}
v^2_0 & v^2_1 &\ldots &v^2_{d-k}
\end{bmatrix}
\]
is
\[
\mathbf{v}^1+\mathbf{v}^2=\begin{bmatrix}
v^1_0+v^2_0 & v^1_1+v^2_1 &\ldots &v^1_{d-k}+v^2_{d-k}
\end{bmatrix}.
\]
For $h\in \{0,1,\ldots,n-1\}\setminus \{t,t+1, \ldots,2t-1\}$, column $h$ stores $t$ vectors ($t^2$ polynomials)
\[
\mathbf{v}^0_{h},\mathbf{v}^1_{h},\ldots,\mathbf{v}^{d-k}_{h}.
\]
For $\ell=0,1,\ldots,d-k$ and $h=t,t+1,\ldots,n-1$, we have that
\[
\mathbf{v}^\ell_{h}=\begin{bmatrix}
a_{h,\ell t}(x) & a_{h,\ell t+1}(x) & \cdots & a_{h,(\ell+1) t-1}(x)
\end{bmatrix}
\]
according to the first transformation.
Table \ref{table:thm-proof2} shows the storage of the first $2t$ columns of $\mathsf{EVENODD}_2$ by \eqref{eq:trans-vectors}. We show in the next theorem that the optimal repair access property of the first $d-k+1$ columns of $\mathsf{EVENODD}_1$ code is maintained in $\mathsf{EVENODD}_2$.

\begin{table*}[!t]
\scriptsize
\caption{The storage of the first $2t$ columns of $\mathsf{EVENODD}_2$, where $t=d-k+1$.}
\begin{center}
\begin{tabular}{|c|c|c|c|}
\hline
Column 0 & Column 1 & $\cdots$  & Column $d-k$  \\
\hline
$a_{0,0}(x)$& $a_{1,0}(x)+a_{0,1}(x)$& $\cdots$& $a_{d-k,0}(x)+a_{0,d-k}(x)$  \\
$a_{0,1}(x)+(1+x^e)a_{1,0}(x)$& $a_{1,1}(x)$& $\cdots$& $a_{d-k,1}(x)+a_{1,d-k}(x)$ \\
$\vdots$&$\vdots$& $\ddots$& $\vdots$  \\
$a_{0,d-k}(x)+(1+x^e)a_{d-k,0}(x)$& $a_{1,d-k}(x)+(1+x^e)a_{d-k,1}(x)$&$\cdots$& $a_{d-k,d-k}(x)$  \\
\hline
$\vdots$&$\vdots$& $\vdots$& $\vdots$  \\\hline
$a_{0,(d-k)t}(x)$& $a_{1,(d-k)t}(x)+a_{0,(d-k)t+1}(x)$& $\cdots$& $a_{d-k,(d-k)t}(x)+a_{0,t^2-1}(x)$  \\
$a_{0,(d-k)t+1}(x)+(1+x^e)a_{1,(d-k)t}(x)$& $a_{1,(d-k)t+1}(x)$& $\cdots$& $a_{d-k,(d-k)t+1}(x)+a_{1,t^2-1}(x)$ \\
$\vdots$&$\vdots$& $\ddots$& $\vdots$  \\
$a_{0,t^2-1}(x)+(1+x^e)a_{d-k,(d-k)t}(x)$& $a_{1,t^2-1}(x)+(1+x^e)a_{d-k,(d-k)t+1}(x)$&$\cdots$& $a_{d-k,t^2-1}(x)$  \\
\hline
\end{tabular}
\begin{tabular}{|c|c|c|c|c|c|c|}
\hline
Column $t$ & Column $t+1$ & $\cdots$ & Column $2t-1$\\
\hline
 $a_{t,0}(x)$ &  $a_{t+1,0}(x)+a_{t,t}(x)$ & $\cdots$ &  $a_{2t-1,0}(x)+a_{t,(d-k)t}(x)$\\
 $a_{t,1}(x)$ &   $a_{t+1,1}(x)+a_{t,t+1}(x)$  & $\cdots$ &  $a_{2t-1,1}(x)+a_{t,(d-k)t+1}(x)$\\
 $\vdots$ & $\vdots$  & $\cdots$&$\vdots$ \\
 $a_{t,d-k}(x)$ &   $a_{t+1,d-k}(x)+a_{t,2t-1}(x)$ & $\cdots$  &  $a_{2t-1,d-k}(x)+a_{t,t^2-1}(x)$\\
\hline
  $a_{t,t}(x)+(1+x^e)a_{t+1,0}(x)$ & $a_{t+1,t}(x)$ & $\cdots$ & $a_{2t-1,t}(x)+a_{t+1,(d-k)t}(x)$\\
  $a_{t,t+1}(x)+(1+x^e)a_{t+1,1}(x)$ & $a_{t+1,t+1}(x)$  & $\cdots$& $a_{2t-1,t+1}(x)+a_{t+1,(d-k)t+1}(x)$\\
  $\vdots$ & $\vdots$ & $\cdots$ & $\vdots$\\
  $a_{t,2t-1}(x)+(1+x^e)a_{t+1,d-k}(x)$ & $a_{t+1,2t-1}(x)$  & $\cdots$& $a_{2t-1,2t-1}(x)+a_{t+1,t^2-1}(x)$\\
\hline
  $\vdots$ & $\vdots$ & $\cdots$ & $\vdots$\\
\hline
  $a_{t,(d-k)t}(x)+(1+x^e)a_{2t-1,0}(x)$&  $a_{t+1,(d-k)t}(x)+(1+x^e)a_{2t-1,t}(x)$ & $\cdots$& $a_{2t-1,(d-k)t}(x)$\\
  $a_{t,(d-k)t+1}(x)+(1+x^e)a_{2t-1,1}(x)$&  $a_{t+1,(d-k)t+1}(x)+(1+x^e)a_{2t-1,t+1}(x)$ & $\cdots$& $a_{2t-1,(d-k)t+1}(x)$\\
  $\vdots$& $\vdots$ & $\cdots$& $\vdots$\\
  $a_{t,t^2-1}(x)+(1+x^e)a_{2t-1,d-k}(x)$&  $a_{t+1,t^2-1}(x)+(1+x^e)a_{2t-1,2t-1}(x)$ & $\cdots$& $a_{2t-1,t^2-1}(x)$\\
\hline
\end{tabular}
\end{center}
\label{table:thm-proof2}
\end{table*}

\begin{theorem}
The repair access of column $i$ of $\mathsf{EVENODD}_2$ code for $i=0,1,\ldots,2d-2k+1$ is optimal.
\label{thm:rep1}
\end{theorem}
\begin{proof}
By Theorem \ref{thm:rep}, we can repair $t$ vectors ($t^2$ polynomials) in column $i$ for $i=t,t+1\ldots,2t-1$ by downloading $k$ vectors ($kt$ polynomials)
\begin{equation}
\mathbf{v}^{i-t}_{h_1},\mathbf{v}^{i-t}_{h_2},\ldots,\mathbf{v}^{i-t}_{h_k}
\label{eq:thm3-1}
\end{equation}
from columns $h_j$ with $j=0,1,\ldots,k-1$, where $h_j=j$ for $j=0,1,\ldots,t-1$ and $h_j\in \{2t,\ldots,k+r-1\}$ for $j=t,t+1,\ldots,k-1$, and the following $d-k$ vectors ($(d-k)t$ polynomials)
\begin{equation}
\begin{array}{ll}
\mathbf{v}^{i-t}_t+(1+x^e)\mathbf{v}^{0}_i,\ldots,\mathbf{v}^{i-t}_{i-1}+(1+x^e)\mathbf{v}^{i-t-1}_i,
\mathbf{v}^{i-t}_{i+1}+\mathbf{v}^{i-t+1}_i,\ldots,\mathbf{v}^{i-t}_{2t-1}+\mathbf{v}^{t-1}_i.
\end{array}
\label{eq:thm3-2}
\end{equation}
Specifically, we can first compute vectors
\[
\mathbf{v}^{i-t}_{t},\mathbf{v}^{i-t}_{t+1},\ldots,\mathbf{v}^{i-t}_{2t-1}
\]
from $k$ vectors in \eqref{eq:thm3-1}, and then recover the $t$ vectors in column $i$ with the above $t$ vectors and the downloaded $d-k$ vectors in \eqref{eq:thm3-2}.

Consider the repair of column $i$ with $i=0,1,\ldots,d-k$.
We can repair $t^2$ polynomials in column $i$ by downloading
$(2t-1)t$ polynomials from $2t-1$ columns $0,1,\ldots,i-1,i+1,\ldots,t-1,t,t+1,\ldots,2t-1$ in rows $i+1,i+1+t,\ldots,i+1+(t-1)t$, and $(d-2t+1)t$ polynomials from
$d-2t+1$ columns $h_1,\ldots,h_{d-2t+1}$ in rows $i+1,i+1+t,\ldots,i+1+(t-1)t$ with indices $2t\leq h_1<\ldots <h_{d-2t+1}\leq n-1$. Note that the $t^2$ polynomials downloaded from columns $t$ to $2t-1$ are in \eqref{eq:thm3-3} (in the next page).
\begin{figure*}[!t]
\begin{equation}
\small
\begin{bmatrix}
a_{t,i}(x) &a_{t+1,i}(x)+a_{t,t+i}(x)&\cdots &a_{2t-1,i}(x)+a_{t,(d-k)t+i}(x)\\
a_{t,t+i}(x)+(1+x^e)a_{t+1,i}(x) &a_{t+1,t+i}(x)&\cdots &a_{2t-1,t+i}(x)+a_{t+1,(d-k)t+i}(x)\\
\vdots & \vdots & \ddots & \vdots \\
a_{t,(d-k)t+i}(x)+(1+x^e)a_{2t-1,i}(x) &a_{t+1,(d-k)t+i}(x)+(1+x^e)a_{2t-1,t+i}(x)&\cdots &a_{2t-1,(d-k)t+i}(x)\\
\end{bmatrix}.
\label{eq:thm3-3}
\end{equation}
\end{figure*}
We can compute $a_{t+1,i}(x)$ and $a_{t,t+i}(x)$ from $a_{t+1,i}(x)+a_{t,t+i}(x)$ and $a_{t,t+i}(x)+(1+x^e)a_{t+1,i}(x)$ by $$\frac{(a_{t,t+i}(x)+(1+x^e)a_{t+1,i}(x))-(a_{t+1,i}(x)+a_{t,t+i}(x))}{x^e}$$ and 
\begin{align*}
\frac{(1+x^e)(a_{t+1,i}(x)+a_{t,t+i}(x))-(a_{t,t+i}(x)+(1+x^e)a_{t+1,i}(x))}{x^e}, 
\end{align*}
respectively. Similarly, we can compute $t^2$ polynomials \[
a_{t,\ell t+i}(x),a_{t+1,\ell t+i}(x), \ldots,a_{2t-1,\ell t+i}(x),
\]
from $t^2$ polynomials in \eqref{eq:thm3-3} (in the next page),
where $\ell=0,1,\ldots,d-k$. Together with $(d-2t+1)t$ polynomials
\[
a_{h_1,\ell t+i}(x),a_{h_2,\ell t+i}(x), \ldots, a_{h_{d-2t+1},\ell t+i}(x),
\]
with $\ell=0,1,\ldots,d-k$ downloaded from
$d-2t+1$ columns $h_1,\ldots,h_{d-2t+1}$, we can compute the following $t^2$ polynomials
\[
a_{0,\ell t+i}(x),a_{1,\ell t+i}(x),\ldots, a_{d-k,\ell t+i}(x),
\]
with $\ell=0,1,\ldots,d-k$. Finally, we can recover all $t^2$ polynomials with the above $t^2$ polynomials and the downloaded $(t-1)t$ polynomials from $t-1$ columns $0,1,\ldots,i-1,i+1,\ldots,t-1$.
\end{proof}

Each column of $\mathsf{EVENODD}_1$ stores $d-k+1$ polynomials, we can repair each of the first $d-k+1$ columns by accessing one polynomial from each of the $d$ columns according to Theorem \ref{thm:rep} and the repair access is optimal according to \eqref{eq:rep}. In $\mathsf{EVENODD}_2$, each column has $(d-k+1)^2$ polynomials. According to Theorem \ref{thm:rep1}, the $(d-k+1)^2$ polynomial in each of the first $2(d-k+1)$ columns can be recovered by accessing $d-k+1$ polynomials from each of the $d$ columns and the repair access is optimal according to \eqref{eq:rep}.

\begin{table*}[!t]
\scriptsize
\caption{$\mathsf{EVENODD}_2$ by applying the first transformation twice for EVENODD codes with
$k=4$, $r=2$, $d=5$ and $e=1$.}
\vspace{-8pt}
\begin{center}
\begin{tabular}{|c|c|c|c|c|c|}
\hline
Column 0 & Column 1  & Column 2  & Column 3 & Column 4  & Column 5 \\
\hline
$a_{0,0}(x)$& $a_{1,0}(x)+a_{0,1}(x)$ & $a_{2,0}(x)$& $a_{3,0}(x)+a_{2,2}(x)$ & $a_{4,0}(x)$ & $a_{5,0}(x)$ \\
\hline
$a_{0,1}(x)+(1+x)a_{1,0}(x)$& $a_{1,1}(x)$ &$a_{2,1}(x)$& $a_{3,1}(x)+a_{2,3}(x)$ & $a_{4,1}(x)$ & $a_{5,1}(x)$\\
\hline
\hline
$a_{0,2}(x)$& $a_{1,2}(x)+a_{0,3}(x)$ & $a_{2,2}(x)+(1+x)a_{3,0}(x)$& $a_{3,2}(x)$ & $a_{4,2}(x)$ & $a_{5,2}(x)$ \\
\hline
$a_{0,3}(x)+(1+x)a_{1,2}(x)$& $a_{1,3}(x)$ &$a_{2,3}(x)+(1+x)a_{3,1}(x)$& $a_{3,3}(x)$ & $a_{4,3}(x)$ & $a_{5,3}(x)$\\ \hline
\end{tabular}
\end{center}
\label{table:evenodd2-first}
\end{table*}

Table \ref{table:evenodd2-first} shows the $\mathsf{EVENODD}_2$ by applying the first transformation twice for EVENODD codes with
$k=4$, $r=2$, $d=5$ and $e=1$. We can repair column 0 by downloading the following 10 polynomials
\begin{align*}
&a_{1,0}(x)+a_{0,1}(x), a_{2,0}(x), a_{3,0}(x)+a_{2,2}(x), a_{4,0}(x), \\
&a_{5,0}(x),a_{1,2}(x)+a_{0,3}(x), a_{2,2}(x)+(1+x)a_{3,0}(x), \\
&a_{3,2}(x), a_{4,2}(x), a_{5,2}(x).
\end{align*}
Specifically, we first compute $a_{2,2}(x)$ and $a_{3,0}(x)$ from $a_{3,0}(x)+a_{2,2}(x)$ and $a_{2,2}(x)+(1+x)a_{3,0}(x)$. Then, we can compute $a_{0,0}(x),a_{1,0}(x)$ and $a_{0,2}(x),a_{1,2}(x)$ from
\[
a_{2,0}(x),a_{3,0}(x),a_{4,0}(x),a_{5,0}(x),
\]
and
\[
a_{2,2}(x),a_{3,2}(x),a_{4,2}(x),a_{5,2}(x),
\]
respectively, according to the MDS property of EVENODD codes. Finally, we can recover $a_{0,1}(x)+(1+x)a_{1,0}(x)$ and $a_{0,3}(x)+(1+x)a_{1,2}(x)$ by $(a_{1,0}(x)+a_{0,1}(x))+xa_{1,0}(x)$ and $(a_{1,2}(x)+a_{0,3}(x))+xa_{1,2}(x)$, respectively. According to Theorem \ref{thm:rep1}, we can repair each of the first four columns by downloading 10 polynomials and the repair access is optimal.

\begin{table*}[tbh]
\scriptsize
\caption{Systematic $\mathsf{EVENODD}_2$ code by applying the second transformation twice for EVENODD code with $k=4$, $r=2$, $d=5$ and $e=1$.}
\begin{center}
\begin{tabular}{|c|c|c|c|c|}
\hline
\multicolumn{4}{|c|}{Information columns}     & \multicolumn{1}{l|}{Parity column 0}  \\
\hline
$a_{0,0}(x)$& $a_{1,0}(x)$ & $a_{2,0}(x)$& $a_{3,0}(x)$ & $a_{0,0}(x)+(x^{p-1}a_{1,0}(x)+x^{p-1}a_{0,1}(x))+a_{2,0}(x)+(x^{p-1}a_{3,0}(x)+x^{p-1}a_{2,2}(x))$ \\
\hline
$a_{0,1}(x)$& $a_{1,1}(x)$ & $a_{2,1}(x)$& $a_{3,1}(x)$ & $x^{p-1}a_{0,1}(x)+(1+x^{p-1})a_{1,0}(x)+a_{1,1}(x)+a_{2,1}(x)+x^{p-1}a_{3,1}(x)+x^{p-1}a_{2,3}(x)$ \\
\hline
$a_{0,2}(x)$& $a_{1,2}(x)$ & $a_{2,2}(x)$& $a_{3,2}(x)$ & $a_{0,2}(x)+(x^{p-1}a_{1,2}(x)+x^{p-1}a_{0,3}(x))+x^{p-1}a_{2,2}(x)+(1+x^{p-1})a_{3,0}(x)+a_{3,2}(x)$ \\
\hline
$a_{0,3}(x)$& $a_{1,3}(x)$ & $a_{2,3}(x)$& $a_{3,3}(x)$ & $x^{p-1}a_{0,3}(x)+(1+x^{p-1})a_{1,2}(x)+a_{1,3}(x)+x^{p-1}a_{2,3}(x)+(1+x^{p-1})a_{3,1}(x)+a_{3,3}(x)$ \\
\hline
\end{tabular}
\begin{tabular}{|c|c|c|c|c|}
\hline
Parity column 1\\
\hline
 $a_{0,0}(x)+a_{1,0}(x)+a_{0,1}(x)+x^2a_{2,0}(x)+(x^{2}a_{3,0}(x)+x^{2}a_{2,2}(x))$ \\
\hline
  $x^{p-1}a_{0,1}(x)+(1+x^{p-1})a_{1,0}(x)+xa_{1,1}(x)+x^2a_{2,1}(x)+x^{2}a_{3,1}(x)+x^{2}a_{2,3}(x)$\\
\hline
  $a_{0,2}(x)+(a_{1,2}(x)+a_{0,3}(x))+xa_{2,2}(x)+(x+x^{2})a_{3,0}(x)+x^3a_{3,2}(x)$ \\
\hline
$x^{p-1}a_{0,3}(x)+(1+x^{p-1})a_{1,2}(x)+xa_{1,3}(x)+xa_{2,3}(x)+(x+x^{2})a_{3,1}(x)+x^3a_{3,3}(x)$ \\
\hline
\end{tabular}
\end{center}
\label{table:first2}
\end{table*}

Consider the systematic $\mathsf{EVENODD}_2$ code  with $k=4$, $r=2$, $d=5$ and $e=1$ in Table \ref{table:first2}. The repair access of column $i$ for $i=0,1,2,3$ is optimal. We can recover the four polynomials in column 0 by downloading 10 polynomials
\begin{align*}
&a_{1,0}(x),a_{1,2}(x),a_{2,0}(x),a_{2,2}(x),a_{3,0}(x),a_{3,2}(x),\\
&a_{4,0}(x),a_{4,2}(x),a_{5,0}(x),a_{5,2}(x).
\end{align*}
By subtracting polynomials $a_{1,0}(x),a_{2,0}(x),a_{3,0}(x)$ and $a_{2,2}(x)$ from $a_{4,0}(x)$ and $a_{5,0}(x)$ each, we can obtain two polynomials
\begin{align*}
p_1(x)=&a_{0,0}(x)+x^{p-1}a_{0,1}(x),\\
p_2(x)=&a_{0,0}(x)+a_{0,1}(x).
\end{align*}
Thus, we can recover $a_{0,0}(x)$ and $a_{0,1}(x)$ by $\frac{p_1(x)+x^{p-1}p_2(x)}{1+x^{p-1}}$ and $\frac{p_1(x)+p_2(x)}{1+x^{p-1}}$, respectively. Similarly, we can first obtain
\begin{align*}
p_3(x)=&a_{0,2}(x)+x^{p-1}a_{0,3}(x),\\
p_4(x)=&a_{0,2}(x)+a_{0,3}(x),
\end{align*}
by subtracting polynomials $a_{1,2}(x),a_{2,2}(x),a_{3,2}(x)$ and $a_{3,0}(x)$ from $a_{4,2}(x)$ and $a_{5,2}(x)$ each, and then recover the other two polynomials in column 0 by $\frac{p_3(x)+x^{p-1}p_4(x)}{1+x^{p-1}}$ and $\frac{p_3(x)+p_4(x)}{1+x^{p-1}}$. Column~1 can be recovered by downloading
\begin{align*}
&a_{0,1}(x),a_{0,3}(x),a_{2,1}(x),a_{2,3}(x),a_{3,1}(x),a_{3,3}(x),\\
&a_{4,1}(x),a_{4,3}(x),a_{5,1}(x),a_{5,3}(x).
\end{align*}

The first transformed EVENODD codes also satisfy the above three theorems, as the two transformations are equivalent. The $\mathsf{EVENODD}_1$ codes have optimal repair for each of the first $d-k+1$ columns according to Theorem \ref{thm:rep}. By applying the second transformation for the columns from $d-k+1$ to $2d-2k+1$ of $\mathsf{EVENODD}_1$ codes, we obtain $\mathsf{EVENODD}_2$ codes that have optimal repair for each of the columns from 0 to $2d-2k+1$ according to Theorem \ref{thm:rep1}.
Similarly, we can transform the original EVENODD codes for the columns between $i(d-k+1)$ and $((i+1)(d-k+1)-1)\bmod n$ to obtain the transformed EVENODD codes with optimal repair with the columns between $i(d-k+1)$ and $((i+1)(d-k+1)-1)\bmod n$, where $i=1,2,\ldots,\lceil \frac{n}{d-k+1}\rceil-1$.
The polynomial $1+x^e$ used in the transformation in \eqref{eq:trans1} is called the \emph{encoding coefficient} associated with the transformation. We may replace the encoding coefficient $1+x^e$ by other polynomials in $\mathbb{F}_2[x]/(1+x+\cdots+x^{p-1})$, such as $x^e$, as long as the three theorems in Section \ref{sec:prop} still hold under the specific binary MDS array codes. It is easy to check that the three theorems in Section \ref{sec:prop} hold for EVENODD codes, if we replace the encoding coefficient $1+x^e$ by $x^e$. More generally, we can view the $d-k+1$ polynomials in \eqref{eq:trans1} as $d-k+1$ vectors with length $p-1$ and compute the vectors by the summation of some permutated vectors, as long as the three theorems in Section \ref{sec:prop} hold. With more general transformation, we may combine the transformation and the existing binary MDS array codes with efficient repair for any single information column to obtain the transformed codes that have efficient repair for both information and parity columns. Recall that there are some efficient repair schemes for any information column of RDP \cite{xiang2010}, X-code \cite{xu2014}, EVENODD
\cite{wang2010} and the binary MDS array codes \cite{hou2018c}. In Section \ref{sec:trans-exm}, we will take an example of EVENODD to show how to design the specific transformation for EVENODD to enable optimal repair for each of the parity columns and preserve the efficient repair property for any single information column.
We also give the transformation for the binary MDS array codes \cite{hou2018c} with efficient repair access for any single information column such that the transformed codes have optimal repair access for any single parity column and asymptotically optimal repair access for any single information column in Section~\ref{sec:trans-exm}.

\section{Construction of Multi-layer Transformed EVENODD Codes}

In the section, we first present the construction of multi-layer transformed
EVENODD codes by applying the transformation given in Section
\ref{sec:framework}. We then give a repair algorithm for any single column
with optimal repair access.

\subsection{Construction}
\label{sec:cons}
The codes considered herein have $k$ information columns and $r$ parity columns. For notational convenience, we divide $k$ information columns into $\lceil \frac{k}{d-k+1}\rceil$ \emph{information partitions}, each of the information partitions has $d-k+1$ columns. For $i=1,2,\ldots,\lceil \frac{k}{d-k+1}\rceil-1$, information partition $i$ contains columns between $(i-1)(d-k+1)$ and $(i(d-k+1)-1)$. Information partition $\lceil \frac{k}{d-k+1}\rceil$ contains the last $d-k+1$ information columns. Similarly, the $r$ parity columns are divided into $\lceil \frac{r}{d-k+1}\rceil$ \emph{parity partitions} and parity column $i$ contains parity columns between $(i-1)(d-k+1)$ and $(i(d-k+1)-1)$ for $i=1,2,\ldots,\lceil \frac{r}{d-k+1}\rceil-1$. Parity partition $\lceil \frac{r}{d-k+1}\rceil$ contains the last $d-k+1$ parity columns. Therefore, we obtain $\lceil \frac{k}{d-k+1}\rceil+\lceil \frac{r}{d-k+1}\rceil$ partitions, contains $\lceil \frac{k}{d-k+1}\rceil$ information partitions and $\lceil \frac{r}{d-k+1}\rceil$ parity partitions. We label the index of the partitions from 1 to $\lceil \frac{k}{d-k+1}\rceil+\lceil \frac{r}{d-k+1}\rceil$. The construction is given in the following.

By applying the transformation for the first information partition (the first $d-k+1$ columns) of EVENODD code, we can obtain $\mathsf{EVENODD}_1$ with each column having $d-k+1$ polynomials, such that $\mathsf{EVENODD}_1$ is MDS according to Theorem \ref{thm:mds} and has optimal repair bandwidth for the first $d-k+1$ information columns according to Theorem \ref{thm:rep}. By applying the transformation for the second partition (columns between $d-k+1$ and $2(d-k+1)-1$) of $\mathsf{EVENODD}_1$, we obtain $\mathsf{EVENODD}_2$ with each column having $(d-k+1)^2$ polynomials that is MDS code according to Theorem \ref{thm:mds} and has optimal repair bandwidth for the first $2(d-k+1)$ information columns according to Theorem~\ref{thm:rep} and Theorem~\ref{thm:rep1}.

For $j=1,2,\ldots,\lceil \frac{k}{d-k+1}\rceil-1$, by recursively applying the transformation for information partition $i+1$ of $\mathsf{EVENODD}_j$ code, we can obtain $\mathsf{EVENODD}_{\lceil \frac{k}{d-k+1}\rceil}$.
Specifically, we can obtain $\mathsf{EVENODD}_{j+1}$ by applying the transformation for $\mathsf{EVENODD}_{j}$ as follows, where $j=1,2,\ldots,\lceil \frac{k}{d-k+1}\rceil-1$. We generate $d-k+1$ instances of the code $\mathsf{EVENODD}_j$ and view the $(d-k+1)^j$ polynomials stored in each column of $\mathsf{EVENODD}_j$ as a vector. For $\ell=0,1,\ldots,d-k$ and $h=0,1,\ldots,n-1$, denote $(d-k+1)^j$ polynomials stored in column $h$ of instance $\ell$ of $\mathsf{EVENODD}_j$ as the vector $\mathbf{v}^\ell_h$. For $i=0,1,\ldots,d-k$, column $j(d-k+1)+i$ of $\mathsf{EVENODD}_{j+1}$ stores the following $d-k+1$ vectors ($(d-k+1)^{j+1}$ polynomials)
\begin{align*}
\begin{array}{ll}
& \mathbf{v}^0_{(d-k+1)+i}+\mathbf{v}^i_{(d-k+1)},\\
& \mathbf{v}^1_{(d-k+1)+i}+\mathbf{v}^i_{(d-k+1)+1},\ldots,\\
& \mathbf{v}^{i-1}_{(d-k+1)+i}+\mathbf{v}^i_{(d-k+1)+i-1},\\
& \mathbf{v}^i_{(d-k+1)+i},\\
& \mathbf{v}^{i+1}_{(d-k+1)+i}+(1+x^e)\mathbf{v}^i_{(d-k+1)+i+1},\\
& \mathbf{v}^{i+2}_{(d-k+1)+i}+(1+x^e)\mathbf{v}^i_{(d-k+1)+i+2},\ldots,\\
& \mathbf{v}^{d-k}_{(d-k+1)+i}+(1+x^e)\mathbf{v}^i_{(d-k+1)+d-k},
\end{array}
\end{align*}
where $1\leq e\leq p-1$.
For $h\in \{0,1,\ldots,n-1\}\setminus \{j(d-k+1),j(d-k+1)+1, \ldots,(j+1)(d-k+1)-1\}$, column $h$ stores $d-k+1$ vectors ($(d-k+1)^{j+1}$ polynomials)
\[
\mathbf{v}^0_{h},\mathbf{v}^1_{h},\ldots,\mathbf{v}^{d-k}_{h}.
\]
Note that the transformation from $\mathsf{EVENODD}_{j}$ to $\mathsf{EVENODD}_{j+1}$ is the same as the transformation from $\mathsf{EVENODD}_{1}$ to $\mathsf{EVENODD}_{2}$, and we can show that the optimal repair property of the first $j(d-k+1)$ columns of $\mathsf{EVENODD}_{j}$ is maintained in $\mathsf{EVENODD}_{j+1}$ by the similar proof of Theorem \ref{thm:rep1}.
Therefore, we can obtain the code $\mathsf{EVENODD}_{\lceil \frac{k}{d-k+1}\rceil}$ and each column has $(d-k+1)^{\lceil \frac{k}{d-k+1}\rceil}$ polynomials. $\mathsf{EVENODD}_{\lceil \frac{k}{d-k+1}\rceil}$ is MDS code according to Theorem \ref{thm:mds} and we can repair each of the first $k$ columns by downloading $(d-k+1)^{\lceil \frac{k}{d-k+1}\rceil-1}$ polynomials from each of the chosen $d$ columns, the repair bandwidth of each of the first $k$ columns is optimal according to Theorem \ref{thm:rep} and Theorem \ref{thm:rep1}.
To obtain optimal repair bandwidth for parity columns, we can transform $\mathsf{EVENODD}_{\lceil \frac{k}{d-k+1}\rceil}$ code $\lceil \frac{r}{d-k+1}\rceil$ times into $\mathsf{EVENODD}_{\lceil \frac{k}{d-k+1}\rceil+\lceil \frac{r}{d-k+1}\rceil}$ code by the first transformation. In $\mathsf{EVENODD}_{\lceil \frac{k}{d-k+1}\rceil+\lceil \frac{r}{d-k+1}\rceil}$ code, each column has $(d-k+1)^{\lceil \frac{k}{d-k+1}\rceil+\lceil \frac{r}{d-k+1}\rceil}$ polynomials and we can repair each column by downloading $(d-k+1)^{\lceil \frac{k}{d-k+1}\rceil+\lceil \frac{r}{d-k+1}\rceil-1}$ polynomials from each of the chosen $d$ columns. Therefore, $\mathsf{EVENODD}_{\lceil \frac{k}{d-k+1}\rceil+\lceil \frac{r}{d-k+1}\rceil}$ code has optimal repair bandwidth for any column.

\subsection{Repair Algorithm}

\begin{algorithm}
    \caption{Algorithm of repairing a single failed column $f$, where $0\leq f\leq k+r-1$.}
    \label{alg:A1}
\begin{algorithmic}[1]
\STATE {The column $f$ is failed.}
    \IF {$f\in\{0, 1,\ldots, k-1\}$, denote $f=(d-k+1)m_f+r_f$, where $m_f,r_f$ are two integers with $0\leq m_f$ and $0\leq r_f \leq d-k$.}
            \STATE {Repair the polynomials in column $f$ by downloading $a_{h_1,\ell}(x),a_{h_2,\ell}(x),\ldots,a_{h_d,\ell}(x)$, for $\ell\bmod (d-k+1)^{m_f+1}\in \{r_f\cdot (d-k+1)^{m_f},r_f\cdot (d-k+1)^{m_f}+1,\ldots,(r_f+1)\cdot (d-k+1)^{m_f}-1\}$, where $\{h_1,h_2,\ldots,h_{d-k}\}=\{f-m_f,f-m_f+1,\ldots,f-1,f+1,\ldots,f-m_f+d-k\}$ and columns $h_{d-k+1},\ldots,h_{d}$ are chosen as follows. For $i=d-k+1,\ldots,d$, if $h_i$ belongs to a partition $\ell$ with $\ell>m_f+1$, then all $d-k+1$ columns of the partition $\ell$ are in $\{h_{d-k+1},\ldots,h_{d}\}$.
}
    \ENDIF
            \RETURN
     \IF {$f\in\{k,k+1\ldots,k+r-1\}$, denote $f-k=(d-k+1)m_f+r_f$, where $m_f,r_f$ are two integers with $0\leq m_f$ and $0\leq r_f \leq d-k$.}
            \STATE {Repair the polynomials in column $f$ by downloading $a_{h_1,\ell}(x),a_{h_2,\ell}(x),\ldots,a_{h_d,\ell}(x)$, for $\ell\bmod (d-k+1)^{\lceil \frac{k}{d-k+1} \rceil + m_f+1}\in \{r_f\cdot (d-k+1)^{\lceil \frac{k}{d-k+1} \rceil+m_f},r_f\cdot (d-k+1)^{\lceil \frac{k}{d-k+1} \rceil+m_f}+1,\ldots,(r_f+1)\cdot (d-k+1)^{\lceil \frac{k}{d-k+1} \rceil+m_f}-1\}$, where $\{h_1,h_2,\ldots,h_{d-k}\}=\{f-m_f,f-m_f+1,\ldots,f-1,f+1,\ldots,f-m_f+d-k\}$ and columns $h_{d-k+1},\ldots,h_{d}$ are chosen as follows. For $i=d-k+1,\ldots,d$, if $h_i$ belongs to a partition $\ell$ with $\ell>m_f+1$, then all $d-k+1$ columns of the partition $\ell$ are in $\{h_{d-k+1},\ldots,h_{d}\}$.}
    \ENDIF
            \RETURN
\end{algorithmic}
\end{algorithm}

In the following, we present the repair algorithm for a single column failure
that is stated in Algorithm~\ref{alg:A1}. Note that there is a requirement
when choosing the $d$ helper columns in Algorithm~\ref{alg:A1}. We show in the
next lemma that we can always choose the $d$ helper columns that satisfy the
requirement for all $f$.

\begin{lemma}
For $f=0,1,\ldots,k+r-1$, column $f$ belongs to partition $m_f+1$. The first $d-k$ helper columns are chosen to be the other surviving columns of partition $m_f+1$, and the other $k$ helper columns $h_{i}$ for $i=d-k+1,\ldots,d$ satisfy that if $h_i$ belongs to a partition $\ell$ with $\ell>m_f+1$, then all $d-k+1$ columns of the partition $\ell$ are in $\{h_{d-k+1},\ldots,h_{d}\}$.
\label{lm:alg1}
\end{lemma}
\begin{proof}
We first consider the information failure, i.e., $0\leq f\leq k-1$. The case
of $k\leq f\leq k+r-1$ can be proven similarly.

If $k$ is a multiple of $d-k+1$, then we can choose $k$ helper columns $h_{i}$ for $i=d-k+1,\ldots,d$ as all the columns of any $k/(d-k+1)$ partitions, except partition $m_f+1$. If $k$ is not a multiple of $d-k+1$, we can divide the proof into two cases: $m_f=0$ and $m_f>0$. When $m_f=0$, We can choose $k$ helper columns $h_{i}$ for $i=d-k+1,\ldots,d$ as all the columns of information partitions $\lceil \frac{k}{d-k+1}\rceil-1$ and $\lceil \frac{k}{d-k+1}\rceil$, and any other $\lceil \frac{k}{d-k+1}\rceil-2$ partitions except partition $m_f+1$. When $m_f>0$, we can choose $k$ helper columns $h_{i}$ for $i=d-k+1,\ldots,d$ as all the columns of $\alpha$ partitions except partition $m_f+1$ and $\beta$ columns that belong to information partition $\ell$ with $1\leq \ell\leq m_f$, where $(d-k+1)\alpha+\beta=k$. This completes the proof.
\end{proof}

We first consider the repair algorithm of information column $f$, i.e., $0\leq f\leq k-1$. There exist two integers $m_f$ and $r_f$ such that $f=(d-k+1)m_f+r_f$, where $0\leq m_f$ and $0\leq r_f \leq d-k$.

Note that $\mathsf{EVENODD}_{\lceil \frac{k}{d-k+1}\rceil+\lceil \frac{r}{d-k+1}\rceil}$ code is transformed from EVENODD code for $\lceil \frac{k}{d-k+1}\rceil+\lceil \frac{r}{d-k+1}\rceil$ times. The optimal repair of columns in partition $i$ is enabled by the $i$-th transformation, where $i=1,2,\ldots,\lceil \frac{k}{d-k+1}\rceil+\lceil \frac{r}{d-k+1}\rceil$. According to Theorem \ref{thm:rep1}, the optimal repair property of columns in partition $i$ of $\mathsf{EVENODD}_{i}$ is preserved in $\mathsf{EVENODD}_{i+1}$ (also in $\mathsf{EVENODD}_{\lceil \frac{k}{d-k+1}\rceil+\lceil \frac{r}{d-k+1}\rceil}$) for $i=1,2,\ldots,\lceil \frac{k}{d-k+1}\rceil-1$ if either all $d-k+1$ columns of partition $i+1$ are chosen as helper columns or all $d-k+1$ columns of partition $i+1$ are not chosen as helper columns. In addition, the other $d-k$ surviving columns of partition $i$ are required to recover the failed column in partition $i$. Therefore, the $d$ helper columns of the failed column $f$ are comprised of $d-k$ columns in partition $m_f+1$, and other $k$ columns. If a column of partition $\ell$ with $\ell>m_f+1$ is chosen as helper column, then all $d-k+1$ columns of partition $\ell$ are chosen as the helper columns. By Lemma \ref{lm:alg1}, we can always find the $d$ helper columns that satisfy the requirement in Algorithm \ref{alg:A1}.

We can recover column $f$ of $\mathsf{EVENODD}_{m_f+1}$ by downloading polynomials $a_{h_1,\ell}(x),a_{h_2,\ell}(x),\ldots,a_{h_d,\ell}(x)$ from the chosen $d$ helper columns, for $\ell \in \{r_f\cdot (d-k+1)^{m_f},r_f\cdot (d-k+1)^{m_f}+1,\ldots,2r_f\cdot (d-k+1)^{m_f}-1\}$.
Because the optimal repair algorithm of column $f$ of $\mathsf{EVENODD}_{m_f+1}$ is preserved in $\mathsf{EVENODD}_{\lceil \frac{k}{d-k+1}\rceil+\lceil \frac{r}{d-k+1}\rceil}$ and the number of polynomials of $\mathsf{EVENODD}_{\lceil \frac{k}{d-k+1}\rceil+\lceil \frac{r}{d-k+1}\rceil}$ is extended to be $(d-k+1)^{\lceil \frac{k}{d-k+1}\rceil+\lceil \frac{r}{d-k+1}\rceil}$ that is a multiple of $(d-k+1)^{m_f+1}$ (the number of polynomials of $\mathsf{EVENODD}_{m_f+1}$). Therefore, we can recover column $f$ by step 3 in Algorithm~\ref{alg:A1}. It can be counted that the number of polynomials that are downloaded to recover column $f$ is $d(d-k+1)^{\lceil \frac{k}{d-k+1}\rceil+\lceil \frac{r}{d-k+1}\rceil}$, which is optimal by \eqref{eq:rep}.

When $f=k,k+1,\ldots,k+r-1$, the repair algorithm of column $f$ is similar to the repair algorithm of an information column. The only difference is that the optimal repair property of parity column is enabled by the first transformation, while the optimal repair property of information column is enabled by the second transformation.

\subsection{Decoding Method}
We present the decoding method of any $\rho\leq r$ erasures for $\mathsf{EVENODD}_{\lceil \frac{k}{d-k+1}\rceil+\lceil \frac{r}{d-k+1}\rceil}$ code. Suppose that $\gamma$ information columns $a_{1},\ldots,a_{\gamma}$ and $\delta$ parity columns $b_{1},\ldots,b_{\delta}$ are erased with $0\leq a_{1}<\ldots<a_{\gamma}\leq k-1$ and $0\leq b_{1}<\ldots<b_{\delta}\leq r-1$, where $k> \gamma > 0$, $r\geq \delta \geq 0$ and $\gamma+\delta=\rho\leq r$. Let $$\mathcal{A}:=\{0,1,\ldots,k-1\}\setminus \{a_{1},a_{2},\ldots,a_{\gamma}\}$$
be a set of indices of the available information columns, and let
$$\mathcal{B}:=\{0,1,\ldots,r-1\}\setminus \{b_{1},b_{2},\ldots,b_{\delta}\}$$
be a set of indices of the available parity columns.
We want to first recover the lost information columns by reading $k-\gamma$ information columns with indices $s_1,s_2,\ldots, s_{k-\gamma}\in\mathcal{A}$, and $\gamma$ parity columns with indices $c_1, c_2,\ldots,c_\gamma \in \mathcal{B}$, and then recover the failure parity column by multiplying the corresponding encoding vector and the information polynomials.

In the construction, $k+r$ columns are divided into $\lceil \frac{k}{d-k+1}\rceil+\lceil \frac{r}{d-k+1}\rceil$ partitions with each partition contains $d-k+1$ columns. Columns in each partition are enabled optimal repair access by recursively applying a transformation for each partition. In the decoding procedure, we need to first return to the original EVENODD codes and then decode the failed information polynomials for some rows recursively. The decoding procedure is briefly described as follows.

For $i=1,2,\ldots,\gamma$, there exist two integers $m_{c_i}$ and $r_{c_i}$ such that $c_i=(d-k+1)m_{c_i}+r_{c_i}$, where $0\leq m_{c_i}$ and $0\leq r_{c_i} \leq d-k$. Similarly, we have $s_i=(d-k+1)m_{s_i}+r_{s_i}$ for $i=1,2,\ldots,k-\gamma$, where $0\leq m_{s_i}$ and $0\leq r_{s_i} \leq d-k$.

We consider the case that $m_{c_1}\neq m_{c_2}\neq \cdots \neq m_{c_\gamma}$.
The parity polynomials are either linear combinations of the corresponding
information polynomials or  the summations of some linear combinations of the
information polynomials. According to Theorem \ref{thm:mds}, there exists at
least one row of the array codes, of which we can first obtain $\gamma$
syndrome polynomials by subtracting $(k-\gamma)(d-k+1)$ information
polynomials from $\gamma$ parity polynomials and then solve the $\gamma$ failed information polynomials by computing the $\gamma\times \gamma$ linear equations of the $\gamma$ syndrome polynomials.
Finally, we can recover the failed information polynomials in other rows recursively by subtracting the known information polynomials from the chosen parity polynomials.

For $1\leq i < j\leq \gamma$, if $m_{c_i}= m_{c_j}$, then we can obtain two parity polynomials of $\mathsf{EVENODD}_{\lceil \frac{k}{d-k+1}\rceil}$ according to the remark at the end of Section \ref{sec:first-trans}. After solving the parity polynomials of $\mathsf{EVENODD}_{\lceil \frac{k}{d-k+1}\rceil+1}$ for all $i,j$ with $m_{c_i}= m_{c_j}$, then we can find at least $(d-k+1)^{\lceil \frac{k}{d-k+1} \rceil + m_{c_1}}$ rows such that all the $\gamma$ parity polynomials in the each of the chosen rows are parity polynomials of $\mathsf{EVENODD}_{\lceil \frac{k}{d-k+1}\rceil}$. By the similar decoding procedure of  $m_{c_1}\neq m_{c_2}\neq \cdots \neq m_{c_\gamma}$, all the failed information polynomials can be solved.

\subsection{Example}

\begin{table*}[tbh]
\caption{Systematic transformed $\mathsf{EVENODD}_{2+1}$ code with $k=4$, $r=2$, $d=5$ and $e=1$.}
\begin{center}
\scriptsize
\begin{tabular}{|c|c|c|c|c|}
\hline
\multicolumn{4}{|c|}{Information columns}     & \multicolumn{1}{l|}{Parity column 0}  \\
\hline
$a_{0,0}(x)$& $a_{1,0}(x)$ & $a_{2,0}(x)$& $a_{3,0}(x)$ & $a_{0,0}(x)+(x^{p-1}a_{1,0}(x)+x^{p-1}a_{0,1}(x))+a_{2,0}(x)+(x^{p-1}a_{3,0}(x)+x^{p-1}a_{2,2}(x))$ \\
\hline
$a_{0,1}(x)$& $a_{1,1}(x)$ & $a_{2,1}(x)$& $a_{3,1}(x)$ & $x^{p-1}a_{0,1}(x)+(1+x^{p-1})a_{1,0}(x)+a_{1,1}(x)+a_{2,1}(x)+x^{p-1}a_{3,1}(x)+x^{p-1}a_{2,3}(x)$ \\
\hline
$a_{0,2}(x)$& $a_{1,2}(x)$ & $a_{2,2}(x)$& $a_{3,2}(x)$ & $a_{0,2}(x)+(x^{p-1}a_{1,2}(x)+x^{p-1}a_{0,3}(x))+x^{p-1}a_{2,2}(x)+(1+x^{p-1})a_{3,0}(x)+a_{3,2}(x)$ \\
\hline
$a_{0,3}(x)$& $a_{1,3}(x)$ & $a_{2,3}(x)$& $a_{3,3}(x)$ & $x^{p-1}a_{0,3}(x)+(1+x^{p-1})a_{1,2}(x)+a_{1,3}(x)+x^{p-1}a_{2,3}(x)+(1+x^{p-1})a_{3,1}(x)+a_{3,3}(x)$ \\
\hline
$a_{0,4}(x)$& $a_{1,4}(x)$ & $a_{2,4}(x)$& $a_{3,4}(x)$ & $a_{0,4}(x)+(x^{p-1}a_{1,4}(x)+x^{p-1}a_{0,5}(x))+a_{2,4}(x)+(x^{p-1}a_{3,4}(x)+x^{p-1}a_{2,6}(x))+$ \\
& & & & $(1+x)(a_{0,0}(x)+a_{1,0}(x)+a_{0,1}(x)+x^2a_{2,0}(x)+(x^{2}a_{3,0}(x)+x^{2}a_{2,2}(x)))$\\
\hline
$a_{0,5}(x)$& $a_{1,5}(x)$ & $a_{2,5}(x)$& $a_{3,5}(x)$ & $x^{p-1}a_{0,5}(x)+(1+x^{p-1})a_{1,4}(x)+a_{1,5}(x)+a_{2,5}(x)+x^{p-1}a_{3,5}(x)+x^{p-1}a_{2,7}(x)+$ \\
& & & & $(1+x)(x^{p-1}a_{0,1}(x)+(1+x^{p-1})a_{1,0}(x)+xa_{1,1}(x)+x^2a_{2,1}(x)+x^{2}a_{3,1}(x)+x^{2}a_{2,3}(x))$ \\
\hline
$a_{0,6}(x)$& $a_{1,6}(x)$ & $a_{2,6}(x)$& $a_{3,6}(x)$ & $a_{0,6}(x)+(x^{p-1}a_{1,6}(x)+x^{p-1}a_{0,7}(x))+x^{p-1}a_{2,6}(x)+(1+x^{p-1})a_{3,4}(x)+a_{3,6}(x)+$ \\
& & & & $(1+x)(a_{0,2}(x)+(a_{1,2}(x)+a_{0,3}(x))+xa_{2,2}(x)+(x+x^{2})a_{3,0}(x)+x^3a_{3,2}(x))$\\
\hline
$a_{0,7}(x)$& $a_{1,7}(x)$ & $a_{2,7}(x)$& $a_{3,7}(x)$ & $x^{p-1}a_{0,7}(x)+(1+x^{p-1})a_{1,6}(x)+a_{1,7}(x)+x^{p-1}a_{2,7}(x)+(1+x^{p-1})a_{3,5}(x)+a_{3,7}(x)+$ \\
& & & & $(1+x)(x^{p-1}a_{0,3}(x)+(1+x^{p-1})a_{1,2}(x)+xa_{1,3}(x)+xa_{2,3}(x)+(x+x^{2})a_{3,1}(x)+x^3a_{3,3}(x))$\\
\hline
\end{tabular}
\begin{tabular}{|c|}
\hline
Parity column 1\\
\hline
 $a_{5,0}(x)=a_{0,0}(x)+a_{1,0}(x)+a_{0,1}(x)+x^2a_{2,0}(x)+(x^{2}a_{3,0}(x)+x^{2}a_{2,2}(x))+$ \\
 $a_{0,4}(x)+(x^{p-1}a_{1,4}(x)+x^{p-1}a_{0,5}(x))+a_{2,4}(x)+(x^{p-1}a_{3,4}(x)+x^{p-1}a_{2,6}(x))$\\
\hline
  $a_{5,1}(x)=x^{p-1}a_{0,1}(x)+(1+x^{p-1})a_{1,0}(x)+xa_{1,1}(x)+x^2a_{2,1}(x)+x^{2}a_{3,1}(x)+x^{2}a_{2,3}(x)+$\\
  $x^{p-1}a_{0,5}(x)+(1+x^{p-1})a_{1,4}(x)+a_{1,5}(x)+a_{2,5}(x)+x^{p-1}a_{3,5}(x)+x^{p-1}a_{2,7}(x)$ \\
\hline
  $a_{5,2}(x)=a_{0,2}(x)+(a_{1,2}(x)+a_{0,3}(x))+xa_{2,2}(x)+(x+x^{2})a_{3,0}(x)+x^3a_{3,2}(x)+$ \\
  $a_{0,6}(x)+(x^{p-1}a_{1,6}(x)+x^{p-1}a_{0,7}(x))+x^{p-1}a_{2,6}(x)+(1+x^{p-1})a_{3,4}(x)+a_{3,6}(x)$  \\
\hline
$a_{5,3}(x)=x^{p-1}a_{0,3}(x)+(1+x^{p-1})a_{1,2}(x)+xa_{1,3}(x)+xa_{2,3}(x)+(x+x^{2})a_{3,1}(x)+x^3a_{3,3}(x)+$ \\
$x^{p-1}a_{0,7}(x)+(1+x^{p-1})a_{1,6}(x)+a_{1,7}(x)+x^{p-1}a_{2,7}(x)+(1+x^{p-1})a_{3,5}(x)+a_{3,7}(x)$\\
\hline
 $a_{5,4}(x)=a_{0,4}(x)+a_{1,4}(x)+a_{0,5}(x)+x^2a_{2,4}(x)+(x^{2}a_{3,4}(x)+x^{2}a_{2,6}(x))$ \\
\hline
  $a_{5,5}(x)=x^{p-1}a_{0,5}(x)+(1+x^{p-1})a_{1,4}(x)+xa_{1,5}(x)+x^2a_{2,5}(x)+x^{2}a_{3,5}(x)+x^{2}a_{2,7}(x)$\\
\hline
  $a_{5,6}(x)=a_{0,6}(x)+(a_{1,6}(x)+a_{0,7}(x))+xa_{2,6}(x)+(x+x^{2})a_{3,4}(x)+x^3a_{3,6}(x)$ \\
\hline
$a_{5,7}(x)=x^{p-1}a_{0,7}(x)+(1+x^{p-1})a_{1,6}(x)+xa_{1,7}(x)+xa_{2,7}(x)+(x+x^{2})a_{3,5}(x)+x^3a_{3,7}(x)$ \\
\hline
\end{tabular}
\end{center}
\label{table:first3}
\end{table*}

We present an example of $k=4$, $r=2$, $d=5$ and $e=1$ to illustrate the main ideas. Table \ref{table:A} shows the systematic transformed $\mathsf{EVENODD}_1$ code and Table \ref{table:first2} shows the systematic transformed $\mathsf{EVENODD}_2$ code. While the systematic transformed $\mathsf{EVENODD}_{2+1}$ code is shown in Table \ref{table:first3}, and we focus on the transformed $\mathsf{EVENODD}_{2+1}$ code in the following.

\subsubsection{Decoding Procedure}
We claim that we can recover all the information polynomials from any four columns. From the first four columns, we can obtain all the information polynomials directly. Suppose that the data collector connects to three information columns and one parity column, say columns 0, 1, 2 and 5. From columns 0, 1 and 2, one can download information polynomials $a_{0,\ell}(x),a_{1,\ell}(x),a_{2,\ell}(x)$ for $\ell=0,1,\ldots,7$ directly. By subtracting the downloaded information polynomials from the parity polynomials $a_{5,\ell}(x)$, we can obtain the following 8 polynomials
\begin{align*}
&x^2a_{3,0}(x)+x^{p-1}a_{3,4}(x),x^2a_{3,1}(x)+x^{p-1}a_{3,5}(x),\\
&(x+x^2)a_{3,0}(x)+x^3a_{3,2}(x)+(1+x^{p-1})a_{3,4}(x)+a_{3,6}(x),\\
&(x+x^2)a_{3,1}(x)+x^3a_{3,3}(x)+(1+x^{p-1})a_{3,5}(x)+a_{3,7}(x),\\
&x^2a_{3,4}(x),x^2a_{3,5}(x),(x+x^2)a_{3,4}(x)+x^3a_{3,6}(x),\\
&(x+x^2)a_{3,5}(x)+x^3a_{3,7}(x).
\end{align*}
It is easy to recover the information polynomials $a_{3,\ell}(x)$ for $\ell=0,1,\ldots,7$ from the above polynomials. The decoding of any three information columns and the first parity column is similar.

Suppose that we want to decode the information polynomials from two information columns and two parity columns, say columns 0, 2, 4 and 5. Denote
\begin{align*}
b_0(x)=&a_{0,4}(x)+(x^{p-1}a_{1,4}(x)+x^{p-1}a_{0,5}(x))+a_{2,4}(x)+\\
&(x^{p-1}a_{3,4}(x)+x^{p-1}a_{2,6}(x)),\\
b_1(x)=&x^{p-1}a_{0,5}(x)+(1+x^{p-1})a_{1,4}(x)+a_{1,5}(x)+\\
&a_{2,5}(x)+x^{p-1}a_{3,5}(x)+x^{p-1}a_{2,7}(x),\\
b_2(x)=&a_{0,6}(x)+(x^{p-1}a_{1,6}(x)+x^{p-1}a_{0,7}(x))+\\
&x^{p-1}a_{2,6}(x)+(1+x^{p-1})a_{3,4}(x)+a_{3,6}(x),\\
b_3(x)=&x^{p-1}a_{0,7}(x)+(1+x^{p-1})a_{1,6}(x)+a_{1,7}(x)+\\
&x^{p-1}a_{2,7}(x)+(1+x^{p-1})a_{3,5}(x)+a_{3,7}(x),
\end{align*}
and
\begin{align*}
c_0(x)=&a_{0,0}(x)+a_{1,0}(x)+a_{0,1}(x)+x^2a_{2,0}(x)+\\
&(x^{2}a_{3,0}(x)+x^{2}a_{2,2}(x)),\\
c_1(x)=&x^{p-1}a_{0,1}(x)+(1+x^{p-1})a_{1,0}(x)+xa_{1,1}(x)+\\
&x^2a_{2,1}(x)+x^{2}a_{3,1}(x)+x^{2}a_{2,3}(x),\\
c_2(x)=&a_{0,2}(x)+(a_{1,2}(x)+a_{0,3}(x))+xa_{2,2}(x)+\\
&(x+x^{2})a_{3,0}(x)+x^3a_{3,2}(x)\\
c_3(x)=&x^{p-1}a_{0,3}(x)+(1+x^{p-1})a_{1,2}(x)+xa_{1,3}(x)+\\
&xa_{2,3}(x)+(x+x^{2})a_{3,1}(x)+x^3a_{3,3}(x).
\end{align*}
We have that
\begin{align*}
a_{4,4}(x)=b_{0}(x)+(1+x)c_0(x),a_{5,0}(x)=b_{0}(x)+c_0(x),\\
a_{4,5}(x)=b_{1}(x)+(1+x)c_1(x),a_{5,1}(x)=b_{1}(x)+c_1(x),\\
a_{4,6}(x)=b_{2}(x)+(1+x)c_2(x),a_{5,2}(x)=b_{2}(x)+c_2(x),\\
a_{4,7}(x)=b_{3}(x)+(1+x)c_3(x),a_{5,3}(x)=b_{3}(x)+c_3(x).
\end{align*}
Therefore, we can solve $c_{\ell}(x)$ by
\begin{align*}
c_{\ell}(x)=x^{p-1}(a_{4,4+\ell}(x)+a_{5,\ell}(x))
\end{align*}
and $b_{\ell}(x)$ by $c_{\ell}(x)+a_{5,\ell}(x)$ for $\ell=0,1,2,3$.
Then, we can subtract the information polynomials $a_{0,\ell}(x),a_{2,\ell}(x)$ for $\ell=0,1,\ldots,7$ from the polynomials $a_{4,\ell}(x),a_{5,4+\ell}(x)$, $b_{\ell}(x),c_{\ell}(x)$ for $\ell=0,1,2,3$, and obtain the polynomials
\begin{align*}
&p_0(x)=x^{p-1}a_{1,0}(x)+x^{p-1}a_{3,0}(x),\\
&p_1(x)=a_{1,0}(x)+x^2a_{3,0}(x),\\
&p_2(x)=(1+x^{p-1})a_{1,0}(x)+a_{1,1}(x)+x^{p-1}a_{3,1}(x),\\
&p_3(x)=(1+x^{p-1})a_{1,0}(x)+xa_{1,1}(x)+x^2a_{3,1}(x),\\
&p_4(x)=x^{p-1}a_{1,2}(x)+(1+x^{p-1})a_{3,0}(x)+a_{3,2}(x), \\
&p_5(x)=a_{1,2}(x)+(x+x^2)a_{3,0}(x)+x^3a_{3,2}(x),\\
&p_6(x)=(1+x^{p-1})a_{1,2}(x)+a_{1,3}(x)+(1+x^{p-1})a_{3,1}(x)+a_{3,3}(x),\\
&p_7(x)=(1+x^{p-1})a_{1,2}(x)+xa_{1,3}(x)+(x+x^2)a_{3,1}(x)+x^3a_{3,3}(x),\\
&p_8(x)=a_{1,4}(x)+x^2a_{3,4}(x),\\
&p_9(x)=x^{p-1}a_{1,4}(x)+x^{p-1}a_{3,4}(x),\\
&p_{10}(x)=(1+x^{p-1})a_{1,4}(x)+xa_{1,5}(x)+x^2a_{3,5}(x),\\
&p_{11}(x)=(1+x^{p-1})a_{1,4}(x)+a_{1,5}(x)+x^{p-1}a_{3,5}(x),\\
&p_{12}(x)=x^{p-1}a_{1,6}(x)+(1+x^{p-1})a_{3,4}(x)+a_{3,6}(x),\\
&p_{13}(x)=a_{1,6}(x)+(x+x^{2})a_{3,4}(x)+x^3a_{3,6}(x),\\
&p_{14}(x)=(1+x^{p-1})a_{1,6}(x)+a_{1,7}+(1+x^{p-1})a_{3,5}(x)+a_{3,7}(x),\\
&p_{15}(x)=(1+x^{p-1})a_{1,6}(x)+xa_{1,7}+(x+x^{2})a_{3,5}(x)+x^3a_{3,7}(x).
\end{align*}
We can first compute
\[
a_{3,0}(x)=\frac{p_0(x)+x^{p-1}p_1(x)}{x+x^{p-1}}
\]
and then compute $a_{1,0}(x)=p_1(x)+x^2a_{3,0}(x)$. The other polynomials $$a_{1,1}(x),a_{3,1}(x);a_{1,2}(x),a_{3,2}(x);a_{1,3}(x),a_{3,3}(x)$$
can be computed by recursively subtracting the known polynomials from
$$p_{2}(x),p_{3}(x);p_{4}(x),p_{5}(x);p_{6}(x),p_{7}(x).$$
Similarly, polynomials
\begin{small}
\begin{align*}
a_{1,4}(x),a_{3,4}(x);a_{1,5}(x),a_{3,5}(x);a_{1,6}(x),a_{3,6}(x);a_{1,7}(x),a_{3,7}(x)
\end{align*}
\end{small}
can be computed from
$$p_{8}(x),p_{9}(x);p_{10}(x),p_{11}(x);p_{12}(x),p_{13}(x);p_{14}(x),p_{15}(x).$$

\subsubsection{Repair Procedure}

Next we show that any one column can be recovered by Algorithm \ref{alg:A1} with optimal repair bandwidth. Suppose that column 4 (parity column 0) is failed, i.e., $f=4$. As $k=4$, $r=2$ and $d=5$, we have $m_f=0$ and $r_f=0$ and all the surviving five columns are selected as helper columns. By step 6 of Algorithm \ref{alg:A1}, we need to download polynomials $$a_{0,\ell}(x),a_{1,\ell}(x),a_{2,\ell}(x),a_{3,\ell}(x),a_{5,\ell}(x)$$ from columns $0,1,2,3,5$ for $\ell=0,1,2,3$ to recover column 4. First, we can directly compute the following four parity polynomials
\begin{align*}
a_{4,0}(x)=&a_{0,0}(x)+(x^{p-1}a_{1,0}(x)+x^{p-1}a_{0,1}(x))+\\
&a_{2,0}(x)+(x^{p-1}a_{3,0}(x)+x^{p-1}a_{2,2}(x)), \\
a_{4,1}(x)=&x^{p-1}a_{0,1}(x)+(1+x^{p-1})a_{1,0}(x)+a_{1,1}(x)+\\
&a_{2,1}(x)+x^{p-1}a_{3,1}(x)+x^{p-1}a_{2,3}(x), \\
a_{4,2}(x)=&a_{0,2}(x)+(x^{p-1}a_{1,2}(x)+x^{p-1}a_{0,3}(x))+\\
&x^{p-1}a_{2,2}(x)+(1+x^{p-1})a_{3,0}(x)+a_{3,2}(x), \\
a_{4,3}(x)=&x^{p-1}a_{0,3}(x)+(1+x^{p-1})a_{1,2}(x)+a_{1,3}(x)+\\
&x^{p-1}a_{2,3}(x)+(1+x^{p-1})a_{3,1}(x)+a_{3,3}(x),
\end{align*}
from the downloaded information polynomials. Then, we can compute the following four polynomials
\begin{align*}
c_0(x)=&a_{0,0}(x)+a_{1,0}(x)+a_{0,1}(x)+x^2a_{2,0}(x)+\\
&(x^{2}a_{3,0}(x)+x^{2}a_{2,2}(x)),\\
c_1(x)=&x^{p-1}a_{0,1}(x)+(1+x^{p-1})a_{1,0}(x)+xa_{1,1}(x)+\\
&x^2a_{2,1}(x)+x^{2}a_{3,1}(x)+x^{2}a_{2,3}(x),\\
c_2(x)=&a_{0,2}(x)+(a_{1,2}(x)+a_{0,3}(x))+xa_{2,2}(x)+\\
&(x+x^{2})a_{3,0}(x)+x^3a_{3,2}(x)\\
c_3(x)=&x^{p-1}a_{0,3}(x)+(1+x^{p-1})a_{1,2}(x)+xa_{1,3}(x)+\\
&xa_{2,3}(x)+(x+x^{2})a_{3,1}(x)+x^3a_{3,3}(x),
\end{align*}
from the downloaded information polynomials. Lastly, we can recover
the polynomials $a_{4,4}(x)$, $a_{4,5}(x)$, $a_{4,6}(x)$ and $a_{4,7}(x)$ by computing
\begin{align*}
a_{4,4}(x)=&a_{5,0}(x)+xc_0(x),\\
a_{4,5}(x)=&a_{5,1}(x)+xc_1(x),\\
a_{4,6}(x)=&a_{5,2}(x)+xc_2(x),\\
a_{4,7}(x)=&a_{5,3}(x)+xc_3(x).
\end{align*}

It can be checked that column 5 can be recovered by downloading $a_{0,\ell}(x)$, $a_{1,\ell}(x)$, $a_{2,\ell}(x)$, $a_{3,\ell}(x)$ and $a_{4,\ell}(x)$ from columns $0,1,2,3,4$ for $\ell=4,5,6,7$ according to Algorithm \ref{alg:A1}. Similarly, we can recover column 2 and column 3 by downloading $a_{0,\ell}(x),a_{1,\ell}(x),a_{3,\ell}(x),a_{4,\ell}(x),a_{5,\ell}(x)$ for $\ell=0,2,4,6$, and $a_{0,\ell}(x),a_{1,\ell}(x),a_{2,\ell}(x),a_{4,\ell}(x),a_{5,\ell}(x)$ for $\ell=2,3,6,7$, respectively. According to Algorithm \ref{alg:A1}, column 0 and column 1 can be recovered by downloading $a_{1,\ell}(x),a_{2,\ell}(x),a_{3,\ell}(x)$, $a_{4,\ell}(x),a_{5,\ell}(x)$ for $\ell=0,2,4,6$, and $a_{0,\ell}(x),a_{2,\ell}(x),a_{3,\ell}(x),a_{4,\ell}(x),a_{5,\ell}(x)$ for $\ell=1,3,5,7$, respectively.

\section{Transformation for Other Binary MDS Array Codes}
\label{sec:trans-exm}
The transformation given in Section \ref{sec:trans} can also be employed in
other binary MDS array codes, such as RDP and codes in
\cite{blomer1999,feng2005,hou2018a,schindelhauer2013,hou2017,hou2018c,hou2018b}.

Specifically, for RDP and codes in
\cite{blomer1999,feng2005,hou2018a,schindelhauer2013}, the transformation is
similar to that of EVENODD in Section~\ref{sec:cons}. We only need to replace
the original EVENODD by the new codes and transform them
for $\lceil \frac{k}{d-k+1}\rceil+\lceil \frac{r}{d-k+1}\rceil$ times. We can
show that the multi-layer transformed codes also
have optimal repair access for all columns, as in $\mathsf{EVENODD}_{\lceil
\frac{k}{d-k+1}\rceil+\lceil \frac{r}{d-k+1}\rceil}$.

For codes with optimal repair access or asymptotically optimal repair access
only for information column, such as codes in
\cite{hou2017,hou2018c,hou2018b,gad2013,pamies2016}, we can transform them for
$\lceil \frac{r}{d-k+1}\rceil$ times to enable optimal repair access for all
$r$ parity columns and the optimal repair property of any information column is
preserved. In the following, we take an example with $k=2$, $r=2$ and $d=3$
of the construction in \cite{hou2018c} to illustrate how to apply the transformation
to obtain the transformed codes with optimal repair access for each of the two parity
columns and asymptotically optimal repair access for each of the two information columns.

\subsection{Transformation for Array Codes in \cite{hou2018c}}

The array code in \cite{hou2018c} is specified by parameters $k$, $r$, $d$, $p$ and $\tau$. Let $k=2$, $r=2$, $d=3$, $p=3$ and $\tau=4$. The array of the example is of size $8\times 4$. The first two columns are information columns that store information bits and the last two columns are parity columns that store parity bits. Let $a_{i,j}$ be the $i$-th bit in column $j$, where $i=0,1,\ldots,7$ and $j=0,1,2,3$. For $j=0,1$, we define four \emph{extra bits} $a_{8,j},a_{9,j},a_{10,j},a_{11,j}$ associated with column $j$ as
\begin{eqnarray}
a_{8,j}=& a_{0,j}+a_{4,j},\nonumber\\
a_{9,j}=& a_{1,j}+a_{5,j},\nonumber\\
a_{10,j}=& a_{2,j}+a_{6,j},\nonumber\\
a_{11,j}=&a_{3,j}+a_{7,j}.
\label{eq:exm}
\end{eqnarray}
Given the information bits, the parity bits $a_{0,2},a_{1,2},\ldots,a_{7,2}$ in column 2 are computed by
\[
a_{i,2}=a_{i,0}+a_{i,1} \text{ for } i=0,1,\ldots,7,
\]
and the parity bits $a_{0,3},a_{1,3},\ldots,a_{7,3}$ in column 3 are computed by
\[
a_{i,3}=a_{i-1,0}+a_{i-2,1} \text{ for } i=0,1,\ldots,7.
\]
Note that all the subscripts in the example are computed by modulo 12. Table \ref{table:exm} shows the example.
Similar to the information column, we also define four \emph{extra bits} $a_{8,j},a_{9,j},a_{10,j},a_{11,j}$ associated with column $j$ as in \eqref{eq:exm}. Note that we only store eight bits $a_{0,j},a_{1,j},\ldots,a_{7,j}$ in column $j$, where $j=0,1,2,3$, and we can compute the extra bits when necessary by \eqref{eq:exm}.

\begin{table}
\caption{An example of the array code in \cite{hou2018c} with $k=2$, $r=2$, $d=3$, $p=3$ and $\tau=4$.}
\begin{center}
\begin{tabular}{|c|c|c|c|}
\hline
Column 0 & Column 1  & Column 2  & Column 3  \\
\hline
$a_{0,0}$& $a_{0,1}$ &  $a_{0,2}=a_{0,0}+a_{0,1}$ & $a_{0,3}=a_{11,0}+a_{10,1}$  \\
\hline
$a_{1,0}$& $a_{1,1}$ &  $a_{1,2}=a_{1,0}+a_{1,1}$ & $a_{1,3}=a_{0,0}+a_{11,1}$  \\
\hline
$a_{2,0}$& $a_{2,1}$ &  $a_{2,2}=a_{2,0}+a_{2,1}$ & $a_{2,3}=a_{1,0}+a_{0,1}$  \\
\hline
$a_{3,0}$& $a_{3,1}$ &  $a_{3,2}=a_{3,0}+a_{3,1}$ & $a_{3,3}=a_{2,0}+a_{1,1}$  \\
\hline
$a_{4,0}$& $a_{4,1}$ &  $a_{4,2}=a_{4,0}+a_{4,1}$ & $a_{4,3}=a_{3,0}+a_{2,1}$  \\
\hline
$a_{5,0}$& $a_{5,1}$ &  $a_{5,2}=a_{5,0}+a_{5,1}$ & $a_{5,3}=a_{4,0}+a_{3,1}$  \\
\hline
$a_{6,0}$& $a_{6,1}$ &  $a_{6,2}=a_{6,0}+a_{6,1}$ & $a_{6,3}=a_{5,0}+a_{4,1}$  \\
\hline
$a_{7,0}$& $a_{7,1}$ &  $a_{7,2}=a_{7,0}+a_{7,1}$ & $a_{7,3}=a_{6,0}+a_{5,1}$  \\
\hline
\end{tabular}
\end{center}
\label{table:exm}
\end{table}

We can repair four bits $a_{0,0},a_{2,0},a_{4,0},a_{6,0}$ in column 0 by
\begin{align*}
a_{0,0}=&a_{0,1}+a_{0,2}, \text{ where } a_{0,2}=a_{0,0}+a_{0,1},\\
a_{2,0}=&a_{2,1}+a_{2,2}, \text{ where } a_{2,2}=a_{2,0}+a_{2,1},\\
a_{4,0}=&a_{4,1}+a_{4,2}, \text{ where } a_{4,2}=a_{4,0}+a_{4,1},\\
a_{6,0}=&a_{6,1}+a_{6,2}, \text{ where } a_{6,2}=a_{6,0}+a_{6,1},
\end{align*}
and the other bits $a_{1,0},a_{3,0},a_{5,0},a_{7,0}$ in column 0 by
\begin{align*}
a_{1,0}=&a_{0,1}+a_{2,3}, \text{ where } a_{2,3}=a_{1,0}+a_{0,1},\\
a_{3,0}=&a_{2,1}+a_{4,3}, \text{ where } a_{4,3}=a_{3,0}+a_{2,1},\\
a_{5,0}=&a_{4,1}+a_{6,3}, \text{ where } a_{6,3}=a_{5,0}+a_{4,1},\\
a_{7,0}=&a_{6,1}+a_{4,3}+a_{0,3}, \text{ where } a_{4,3}=a_{3,0}+a_{2,1} \\
&\text{ and } a_{0,3}=a_{11,0}+a_{10,1}.
\end{align*}
We need to download 12 bits
\[
a_{0,1},a_{2,1},a_{4,1},a_{6,1},a_{0,2},a_{2,2},a_{4,2},a_{6,2},a_{0,3},a_{2,3},a_{4,3},a_{6,3},
\]
to recover the eight bits in column 0. Therefore, the repair access of column 0 is optimal according to \eqref{eq:rep}. Column~1 can be recovered by downloading 14 bits
\begin{align*}
&a_{0,0},a_{1,0},a_{3,0},a_{4,0},a_{5,0},a_{7,0},a_{0,2},a_{1,2},\\
&a_{4,2},a_{5,2},a_{0,3},a_{1,3},a_{4,3},a_{5,3},
\end{align*}
which are two bits more than the optimal repair access.

We argue that we can recover the information bits from any two columns. We can directly obtain the information bits from column 0 and column 1. We can also obtain the information bits from any one information column and any one parity column. For example, if we want to decode the information bits from column 0 and column 2, we can subtract $a_{i,0}$ from $a_{i,2}$ to obtain $a_{i,1}$, for $i=0,1,\ldots,7$. Finally, we can decode the information bits from column 2 and column 3 as follows. We can compute $a_{7,1}$ by
\begin{align*}
a_{7,1}=a_{0,2}+a_{1,2}+a_{2,2}+a_{3,2}+a_{1,3}+a_{2,3}+a_{3,3}+a_{4,3}.
\end{align*}
Similarly, we can compute $a_{8,1},a_{9,1},a_{10,1}$ by
\begin{align*}
a_{8,1}=& a_{1,2}+a_{2,2}+a_{3,2}+a_{4,2}+a_{2,3}+a_{3,3}+a_{4,3}+a_{5,3},\\
a_{9,1}=& a_{2,2}+a_{3,2}+a_{4,2}+a_{5,2}+a_{3,3}+a_{4,3}+a_{5,3}+a_{6,3},\\
a_{10,1}=& a_{3,2}+a_{4,2}+a_{5,2}+a_{6,2}+a_{4,3}+a_{5,3}+a_{6,3}+a_{7,3}.
\end{align*}
Then, we can compute the other information bits iteratively.

Next, we show how to apply the first transformation for the above example to obtain the transformed codes that have optimal repair access for each of columns 2 and 3, while the efficient repair property of each of columns 0 and 1 is also preserved. For $j=0,1,2,3$, we can represent the eight bits in column $j$ and the four extra bits associated with column $j$ by polynomial
$$a_j(x)=a_{0,j}+a_{1,j}x+\ldots+a_{11,j}x^{11}.$$
First, we can generate two instances $a_0(x),a_1(x),a_2(x),a_3(x)$ and $b_0(x),b_1(x),b_2(x),b_3(x)$. Then, we compute polynomials $b_2(x)+x^4a_3(x),a_3(x)+b_2(x)$ over $\mathbb{F}_2[x]/(1+x^{12})$. For $j=0,1$, column $j$ stores the eight coefficients of degrees from zero to seven of the polynomials $a_j(x)$ and $b_{j}(x)$; while column 2 and column 3 stores the eight coefficients of degrees from zero to seven of the polynomials $a_2(x),b_2(x)+x^4a_3(x)$ and $a_3(x)+b_2(x),b_3(x)$, respectively. Table \ref{table:exm1} shows the transformed array codes.

\begin{table*}[!t]
\caption{The transformed array code with $k=2$, $r=2$, $d=3$, $p=3$ and $\tau=4$.}
\begin{center}
\begin{tabular}{|c|c|c|c|}
\hline
Column 0 & Column 1  & Column 2  & Column 3  \\
\hline
$a_{0,0}$& $a_{0,1}$ &  $a_{0,2}=a_{0,0}+a_{0,1}$ & $a_{0,3}+b_{0,2}=a_{11,0}+a_{10,1}+b_{0,0}+b_{0,1}$  \\
$b_{0,0}$& $b_{0,1}$ &  $b_{0,2}+a_{8,3}=b_{0,0}+b_{0,1}+a_{7,0}+a_{6,1}$ & $b_{0,3}=b_{11,0}+b_{10,1}$  \\
\hline
$a_{1,0}$& $a_{1,1}$ &  $a_{1,2}=a_{1,0}+a_{1,1}$ & $a_{1,3}+b_{1,2}=a_{0,0}+a_{11,1}+b_{1,0}+b_{1,1}$  \\
$b_{1,0}$& $b_{1,1}$ &  $b_{1,2}+a_{9,3}=b_{1,0}+b_{1,1}+a_{8,0}+a_{7,1}$ & $b_{1,3}=b_{0,0}+b_{11,1}$  \\
\hline
$a_{2,0}$& $a_{2,1}$ &  $a_{2,2}=a_{2,0}+a_{2,1}$ & $a_{2,3}+b_{2,2}=a_{1,0}+a_{0,1}+b_{2,0}+b_{2,1}$  \\
$b_{2,0}$& $b_{2,1}$ &  $b_{2,2}+a_{10,3}=b_{2,0}+b_{2,1}+a_{9,0}+a_{8,1}$ & $b_{2,3}=b_{1,0}+b_{0,1}$  \\
\hline
$a_{3,0}$& $a_{3,1}$ &  $a_{3,2}=a_{3,0}+a_{3,1}$ & $a_{3,3}+b_{3,2}=a_{2,0}+a_{1,1}+b_{3,0}+b_{3,1}$  \\
$b_{3,0}$& $b_{3,1}$ &  $b_{3,2}+a_{11,3}=b_{3,0}+b_{3,1}+a_{10,0}+a_{9,1}$ & $b_{3,3}=b_{2,0}+b_{1,1}$  \\
\hline
$a_{4,0}$& $a_{4,1}$ &  $a_{4,2}=a_{4,0}+a_{4,1}$ & $a_{4,3}+b_{4,2}=a_{3,0}+a_{2,1}+b_{4,0}+b_{4,1}$  \\
$b_{4,0}$& $b_{4,1}$ &  $b_{4,2}+a_{0,3}=b_{4,0}+b_{4,1}+a_{11,0}+a_{10,1}$ & $b_{4,3}=b_{3,0}+b_{2,1}$  \\
\hline
$a_{5,0}$& $a_{5,1}$ &  $a_{5,2}=a_{5,0}+a_{5,1}$ & $a_{5,3}+b_{5,2}=a_{4,0}+a_{3,1}+b_{5,0}+b_{5,1}$  \\
$b_{5,0}$& $b_{5,1}$ &  $b_{5,2}+a_{1,3}=b_{5,0}+b_{5,1}+a_{0,0}+a_{11,1}$ & $b_{5,3}=b_{4,0}+b_{3,1}$  \\
\hline
$a_{6,0}$& $a_{6,1}$ &  $a_{6,2}=a_{6,0}+a_{6,1}$ & $a_{6,3}+b_{6,2}=a_{5,0}+a_{4,1}+b_{6,0}+b_{6,1}$  \\
$b_{6,0}$& $b_{6,1}$ &  $b_{6,2}+a_{2,3}=b_{6,0}+b_{6,1}+a_{1,0}+a_{0,1}$ & $b_{6,3}=b_{5,0}+b_{4,1}$  \\
\hline
$a_{7,0}$& $a_{7,1}$ &  $a_{7,2}=a_{7,0}+a_{7,1}$ & $a_{7,3}+b_{7,2}=a_{6,0}+a_{5,1}+b_{7,0}+b_{7,1}$  \\
$b_{7,0}$& $b_{7,1}$ &  $b_{7,2}+a_{3,3}=b_{7,0}+b_{7,1}+a_{2,0}+a_{1,1}$ & $b_{7,3}=b_{6,0}+b_{5,1}$  \\
\hline
\end{tabular}
\end{center}
\label{table:exm1}
\end{table*}

First, we show that the efficient repair property of any one information column is preserved in the transformed array codes. Consider the repair method of column 0. We can repair column 0 by downloading 24 bits
\begin{align*}
& a_{0,1},a_{2,1},a_{4,1},a_{6,1},a_{0,2},a_{2,2},a_{4,2},a_{6,2},a_{0,3}+b_{0,2},a_{2,3}+b_{2,2},\\
&a_{4,3}+b_{4,2},a_{6,3}+b_{6,2},b_{0,1},b_{2,1},b_{4,1},b_{6,1},b_{0,2}+a_{8,3},\\
&b_{2,2}+a_{10,3},b_{4,2}+a_{0,3},b_{6,2}+a_{2,3},b_{0,3},b_{2,3},b_{4,3},b_{6,3}.
\end{align*}
Specifically, we can compute the
four bits $a_{0,0}$, $a_{2,0}$, $a_{4,0}$, $a_{6,0}$ in column 0 by
\begin{align*}
a_{0,0}=&a_{0,1}+a_{0,2}, \text{ where } a_{0,2}=a_{0,0}+a_{0,1},\\
a_{2,0}=&a_{2,1}+a_{2,2}, \text{ where } a_{2,2}=a_{2,0}+a_{2,1},\\
a_{4,0}=&a_{4,1}+a_{4,2}, \text{ where } a_{4,2}=a_{4,0}+a_{4,1},\\
a_{6,0}=&a_{6,1}+a_{6,2}, \text{ where } a_{6,2}=a_{6,0}+a_{6,1},
\end{align*}
and $a_{1,0}$, $a_{3,0}$, $a_{5,0}$, $a_{7,0}$ in column 0 by
\begin{align*}
a_{3,0}=&(b_{0,2}+a_{8,3})+(a_{0,3}+b_{0,2})+b_{2,1},\\
a_{5,0}=&(b_{2,2}+a_{10,3})+(a_{2,3}+b_{2,2})+b_{4,1},\\
a_{7,0}=&(b_{4,2}+a_{0,3})+(a_{4,3}+b_{4,2})+b_{6,1},\\
a_{1,0}=&(b_{6,2}+a_{2,3})+(a_{6,3}+b_{6,2})+a_{0,1}+a_{4,1}+a_{5,0}.
\end{align*}
Similarly, we can compute
$b_{1,0},b_{3,0},b_{5,0},b_{7,0}$ in column 0 by
\begin{align*}
b_{1,0}=&b_{0,1}+b_{2,3}, \text{ where } b_{2,3}=b_{1,0}+b_{0,1},\\
b_{3,0}=&b_{2,1}+b_{4,3}, \text{ where } b_{4,3}=b_{3,0}+b_{2,1},\\
b_{5,0}=&b_{4,1}+b_{6,3}, \text{ where } b_{6,3}=b_{5,0}+b_{4,1},\\
b_{7,0}=&b_{6,1}+b_{0,3}+b_{4,3}, \text{ where } b_{4,3}=b_{3,0}+b_{2,1} \\
&\text{ and } b_{0,3}=b_{11,0}+b_{10,1},
\end{align*}
and $b_{0,0},b_{2,0},b_{4,0},b_{6,0}$ in column 0 by
\begin{align*}
b_{0,0}=&(b_{0,2}+a_{8,3})+b_{0,1}+a_{7,0}+a_{6,1},\\
b_{2,0}=&(b_{2,2}+a_{10,3})+b_{2,1}+a_{1,0}+a_{5,0}+a_{0,1}+a_{4,1},\\
b_{4,0}=&(b_{4,2}+a_{0,3})+b_{4,1}+a_{3,0}+a_{7,0}+a_{2,1}+a_{6,1},\\
b_{6,0}=&(b_{6,2}+a_{2,3})+b_{6,1}+a_{1,0}+a_{0,1}.
\end{align*}
The repair access of column 0 is also optimal. Column~1 can be recovered by downloading 28 bits
\begin{align*}
&a_{0,0},a_{1,0},a_{3,0},a_{4,0},a_{5,0},a_{7,0},a_{0,2},a_{1,2},a_{4,2},a_{5,2},a_{0,3}+b_{0,2},\\
&a_{1,3}+b_{1,2},a_{4,3}+b_{4,2},a_{5,3}+b_{5,2},b_{0,0},b_{1,0},b_{3,0},b_{4,0},b_{5,0},b_{7,0},\\
&b_{0,2}+a_{8,3},b_{1,2}+a_{9,3},b_{4,2}+a_{0,3},b_{5,2}+a_{1,3},b_{0,3},b_{1,3},b_{4,3},b_{5,3},
\end{align*}
which are four bits more than the optimal repair access.

According to Theorem \ref{thm:rep1}, the repair access of each of column 2 and column 3 is optimal. We can repair column 2 by downloading 24 bits
\begin{align*}
&a_{0,0},a_{1,0},a_{2,0},a_{3,0},a_{4,0},a_{5,0},a_{6,0},a_{7,0},\\
&a_{0,1},a_{1,1},a_{2,1},a_{3,1},a_{4,1},a_{5,1},a_{6,1},a_{7,1},\\
&a_{0,3}+b_{0,2},a_{1,3}+b_{1,2},a_{2,3}+b_{2,2},a_{3,3}+b_{3,2},\\
&a_{4,3}+b_{4,2},a_{5,3}+b_{5,2},a_{6,3}+b_{6,2},a_{7,3}+b_{7,2},
\end{align*}
from columns 0, 1 and 3, and repair column 3 by downloading 24 bits
\begin{align*}
&b_{0,0},b_{1,0},b_{2,0},b_{3,0},b_{4,0},b_{5,0},b_{6,0},b_{7,0},\\
&b_{0,1},b_{1,1},b_{2,1},b_{3,1},b_{4,1},b_{5,1},b_{6,1},b_{7,1},\\
&b_{0,2}+a_{8,3},b_{1,2}+a_{9,3},b_{2,2}+a_{10,3},b_{3,2}+a_{11,3},\\
&b_{4,2}+a_{0,3},b_{5,2}+a_{1,3},b_{6,2}+a_{2,3},b_{7,2}+a_{3,3},
\end{align*}
from columns 0, 1 and 2.

For the MDS array codes \cite{hou2017,hou2018c,hou2018b} with general parameters $k$ and $r$, each column has $(p-1)\tau$ bits, we can choose $p$ to make the array codes satisfy the MDS property and choose $\tau$ to achieve asymptotically optimal repair bandwidth for each of the $k$ information columns. For $j=0,1,\ldots,k+r-1$, we can represent the $(p-1)\tau$ bits $a_{0,j},a_{1,j},\ldots,a_{(p-1)\tau-1,j}$ in column $j$ and $\tau$ extra bits $a_{(p-1)\tau,j},a_{(p-1)\tau+1,j},\ldots,a_{p\tau-1,j}$ associated with column $j$ by polynomial $a_j(x)=\sum_{i=0}^{p\tau-1}a_{i,j}x^{i}\in \mathbb{F}_2[x]/(1+x^{p\tau})$, where the extra bit $a_{(p-1)\tau+\mu,j}$ with $\mu=0,1,\ldots,\tau-1$ is computed by
\[
a_{(p-1)\tau+\mu,j}=a_{\mu,j}+a_{\tau+\mu,j}+\ldots+a_{(p-2)\tau+\mu,j}.
\]
If we apply the first transformation with encoding coefficient being $x^e$ for the columns from $k$ to $d$ of the MDS array codes, we can obtain the transformed codes with each column containing $d-k+1$ polynomials.
We should carefully choose the encoding coefficient in the transformation in order to make sure that the efficient repair property of any information column of original MDS array codes is preserved in the transformed MDS array codes. In the example with $k=2$, $r=2$ and $d=3$, we choose the encoding coefficient of the transformation to be $x^4$. In fact, the efficient repair property of any information column is also maintained if the encoding coefficient is any polynomial of $\{1+x^4,x^8,1+x^8,x^4+x^8,1+x^4+x^8\}$. However, the efficient repair property of any information column is not maintained if the encoding coefficient is other polynomial in $\mathbb{F}_2[x]/(1+x^{12})$.

Note that the following two properties are the essential reasons to preserve the efficient repair property. First, there is a cyclic structure in the ring $\mathbb{F}_2[x]/(1+x^{12})$. The multiplication of $x^4$ and the polynomial $a_3(x)$ in $\mathbb{F}_2[x]/(1+x^{12})$ can be implemented by cyclicly shifting $4$ positions of $a_3(x)$. Second, the exponent of the encoding coefficient of the transformation, $e=4$ is a multiple of two. Otherwise, the efficient repair property of original array codes is not maintained. In the example, we have $d=k+r-1$, i.e., all the surviving columns are connected to recover a failure column. By applying one transformation for the $r$ parity columns, the transformed array codes will have asymptotically or exactly optimal repair for any single column. However, if $d<k+r-1$, then we may need to employ the transformation for many times, as like the transformation for EVENODD codes. When we apply multiple transformations for the array codes in \cite{hou2017,hou2018c,hou2018b}, we should not only carefully choose the encoding coefficient but also the transformed columns in each transformation, in order to preserve the efficient repair property of the information column.

\subsection{Transformation for EVENODD to Preserve the Efficient Repair Property of Any Information Column}
\label{sec:trans-evenodd}
The number of symbols stored in each column or node is also referred to as the \emph{sub-packetization level}. It is important to have a low sub-packetization level for practical consideration. It is shown in \cite{tamo2014} that the lower bound of sub-packetization of optimal access MDS codes over finite field with $d=n-1$ is $r^{(k-1)/r}$. The sub-packetization of MDS code constructions over finite field with optimal repair access presented in \cite{ye2017,li2017} is $(n-k)^{\lceil \frac{n}{n-k}\rceil}$, for $d=n-1$.
$\epsilon$-minimum storage regenerating ($\epsilon$-MSR) codes are proposed in \cite{rawat2018} to reduce the sub-packetization at a cost of slightly more repair bandwidth. Existing constructions \cite{hou2018c} of binary MDS array codes with $d=k+1$ and asymptotically optimal repair access for any single information column show that the sub-packetization is strictly less than $p\cdot 2^{\frac{k}{r-1}+r-1}$ \cite[Theorem 2]{hou2018c}, where $p$ is a prime and constant number. The existing constructions of MDS codes with asymptotically or exactly optimal repair access have an exponential sub-packetization level. The construction of MDS codes with efficient repair for any column with lower sub-packetization level is attractive. In the following, we take EVENODD codes with $r=2$ as an example to show how to design new transformation to enable optimal repair for any single parity column and the repair access of any single information column is roughly $3/4$ of all the information bits, and thus the sub-packetization level is low.

Consider the example of EVENODD codes with $k=3$, $r=2$ and $p=5$. We have $k=3$ information columns and $r=2$ parity columns. Let $a_{0,j},a_{1,j},a_{2,j},a_{3,j}$ be the four bits in column $j$, where $j=0,1,2,3,4$. Table \ref{table:evenodd1} shows the example.

\begin{table*}[!t]
\caption{The EVENODD code with
$k=3$, $r=2$ and $p=5$.}
\vspace{-8pt}
\begin{center}
\begin{tabular}{|c|c|c|c|c|}
\hline
Column 0 & Column 1  & Column 2  & Column 3 & Column 4 \\
\hline
$a_{0,0}$& $a_{0,1}$ & $a_{0,2}$& $a_{0,3}=a_{0,0}+a_{0,1}+a_{0,2}$ & $a_{0,4}=a_{0,0}+a_{3,2}+(a_{3,1}+a_{2,2})$ \\
\hline
$a_{1,0}$& $a_{1,1}$ & $a_{1,2}$& $a_{1,3}=a_{1,0}+a_{1,1}+a_{1,2}$ & $a_{1,4}=a_{1,0}+a_{0,1}+(a_{3,1}+a_{2,2})$ \\
\hline
$a_{2,0}$& $a_{2,1}$ & $a_{2,2}$& $a_{2,3}=a_{2,0}+a_{2,1}+a_{2,2}$ & $a_{2,4}=a_{2,0}+a_{1,1}+a_{0,2}+(a_{3,1}+a_{2,2})$ \\
\hline
$a_{3,0}$& $a_{3,1}$ & $a_{3,2}$& $a_{3,3}=a_{3,0}+a_{3,1}+a_{3,2}$ & $a_{3,4}=a_{3,0}+a_{2,1}+a_{1,2}+(a_{3,1}+a_{2,2})$ \\
\hline
\end{tabular}
\end{center}
\label{table:evenodd1}
\end{table*}

When we say one information bit is repaired by a parity column (the first parity column or the second parity column), it means that we repair the bit by downloading the parity bit in the parity column that contains the failed information bits and all the information bits that are used to compute the downloaded parity bit except the failed information bit. For example, the bit $a_{0,0}$ is repaired by the first parity column, which means that we download the parity bit $a_{0,3}=a_{0,0}+a_{0,1}+a_{0,2}$ in the first parity column and two information bits $a_{0,1},a_{0,2}$ to recover the information bit $a_{0,0}$.
According to the repair method given in \cite{wang2010}, we can repair two information bits of the failed information column by the first parity column and the other information bits by the second parity column. Consider column 1. We can repair $a_{0,1},a_{1,1}$ by
\begin{align*}
a_{0,1}=& a_{0,0}+a_{0,2}+a_{0,3}, \text{ where } a_{0,3}=a_{0,0}+a_{0,1}+a_{0,2},\\
a_{1,1}=& a_{1,0}+a_{1,2}+a_{1,3}, \text{ where } a_{1,3}=a_{1,0}+a_{1,1}+a_{1,2},
\end{align*}
and repair $a_{2,1},a_{3,1}$ by
\begin{align*}
a_{3,1}=& a_{1,0}+a_{0,1}+a_{2,2}+a_{1,4}, \\
&\text{ where } a_{1,4}=a_{1,0}+a_{0,1}+a_{3,1}+a_{2,2},\\
a_{2,1}=& a_{3,0}+a_{1,2}+a_{3,1}+a_{2,2}+a_{3,4}, \\
&\text{ where } a_{3,4}=a_{3,0}+a_{2,1}+a_{1,2}+a_{3,1}+a_{2,2}.
\end{align*}
We need to download 10 bits to recover column 1, i.e., the repair bandwidth of column 1 is roughly $3/4$ of all $12$ information bits.

Next, we present the transformation for general parameters $k$ and $p$ of EVENODD codes with $r=2$ and $d=k+1$. Each column of the transformed EVENODD codes has $2(p-1)$ bits. The transformed EVENODD codes have optimal repair bandwidth for each parity column and the repair bandwidth of each information column is roughly $3/4$ of all the information bits.

Create two instances of EVENODD codes $a_0(x)$, $a_1(x),\ldots,a_{k+1}(x)$ and $b_0(x),b_1(x),\ldots,b_{k+1}(x)$, where $a_j(x)=\sum_{i=0}^{p-2}a_{i,j}x^i$ and $b_j(x)=\sum_{i=0}^{p-2}b_{i,j}x^i$. The information polynomials are $a_0(x),a_1(x),\ldots,a_{k-1}(x)$ and $b_0(x),b_1(x),\ldots,b_{k-1}(x)$, and the parity polynomials are computed by
\begin{align*}
&\begin{bmatrix}
a_k(x) & a_{k+1}(x)\\
b_k(x) & b_{k+1}(x)
\end{bmatrix}\\
=&\begin{bmatrix}
a_0(x) &a_1(x) & \cdots & a_{k-1}(x)\\
b_0(x) &b_1(x) & \cdots & b_{k-1}(x)
\end{bmatrix}
\begin{bmatrix}
1 &1 \\
1 &x \\
\vdots & \vdots\\
1 &x^{k-1}
\end{bmatrix}.
\end{align*}
Let $\mathbf{a_j}=[a_{0,j},a_{1,j},\ldots,a_{p-2,j}]^T$ and $\mathbf{b_j}=[b_{0,j},b_{1,j},\ldots,b_{p-2,j}]^T$ be the coefficients of polynomials $a_j(x)$ and $b_j(x)$, respectively, where $j=0,1,\ldots,k+1$. Given a column vector $\mathbf{a_0}$, we define
\begin{align*}
\mathbf{a_0^*}=&[a_{1,0},a_{0,0},a_{3,0},a_{2,0},\ldots,a_{p-2,0},a_{p-3,0}]^T,\\
\mathbf{\bar{a}_0}=&[a_{0,0},0,a_{2,0},0,\ldots,a_{p-3,0},0]^T.
\end{align*}
The summation of two column vectors $\mathbf{a_0},\mathbf{a_1}$ is define by
\[
\mathbf{a_0}\oplus\mathbf{a_1}=[a_{0,0}+a_{0,1},a_{1,0}+a_{1,1},\ldots,a_{p-2,0}+a_{p-2,1}]^T.
\]
For example, when $p=5$, we have
\[
\mathbf{a_0}\oplus\mathbf{a_1}=[a_{0,0}+a_{0,1},a_{1,0}+a_{1,1},a_{2,0}+a_{2,1},a_{3,0}+a_{3,1}]^T,
\]
and
\begin{align*}
\mathbf{a_0^*}=&[a_{1,0},a_{0,0},a_{3,0},a_{2,0}]^T,\\
\mathbf{\bar{a}_0}=&[a_{0,0},0,a_{2,0},0]^T.
\end{align*}
For $j=0,1,\ldots,k-1$, column $j$ stores $2(p-1)$ information bits $\mathbf{a_j},\mathbf{b_j}$. The first parity column stores $2(p-1)$ parity bits
\begin{align*}
&\mathbf{a_k}\oplus\mathbf{b_k}=[a_{0,k}+b_{0,k},a_{1,k}+b_{1,k},\ldots,a_{p-2,k}+b_{p-2,k}]^T,\\
&\mathbf{b_{k+1}}=[b_{0,k+1},b_{1,k+1},\ldots,b_{p-2,k+1}]^T,
\end{align*}
and the second parity column stores
$2(p-1)$ parity bits
\begin{align*}
\mathbf{a_k}\oplus\mathbf{\bar{b}_k}\oplus\mathbf{b_k^*}=&[a_{0,k}+b_{0,k}+b_{1,k},a_{1,k}+b_{0,k},\ldots,\\
& a_{p-3,k}+b_{p-3,k}+b_{p-2,k},a_{p-2,k}+b_{p-3,k}]^T,\\
\mathbf{a_{k+1}}=&[a_{0,k+1},a_{1,k+1},\ldots,a_{p-2,k+1}]^T.
\end{align*}

We show that the transformed EVENODD codes satisfy MDS property, i.e., we can retrieve all $2k(p-1)$ information bits from any $k$ columns. Consider the $k$ columns from columns 2 to $k+1$. First, we can compute $(p-1)/2$ bits
$b_{i,k}$
by $(a_{i-1,k}+b_{i-1,k})+(a_{i-1,k}+b_{i-1,k}+b_{i,k})$ for $i=1,3,\ldots,p-2$ and compute $a_{i,k}$ by $b_{i,k}+(a_{i,k}+b_{i,k})$ for $i=1,3,\ldots,p-2$. Then, we can compute $b_{i,k}$ by $a_{i+1,k}+(a_{i+1,k}+b_{i,k})$ for $i=0,2,\ldots,p-3$ and compute $a_{i,k}$ by $b_{i,k}+(a_{i,k}+b_{i,k})$ for $i=0,2,\ldots,p-3$.
Finally, we can obtain the information bits $b_{0,0},b_{1,0},\ldots,b_{p-2,0}$ and $b_{0,1},b_{1,1},\ldots,b_{p-2,1}$ from $b_{0,j},b_{1,j},\ldots,b_{p-2,j}$ for $j=2,3,\ldots,k+1$, as the EVENODD code is MDS code. The information bits $a_{0,0},a_{1,0},\ldots,a_{p-2,0}$ and $a_{0,1},a_{1,1},\ldots,a_{p-2,1}$ can be computed similarly. We can also retrieve all information bits from any $k-1$ information columns and any one parity column. Consider the $k$ columns from column 1 to $k$. We can obtain $k(p-1)$ bits $b_{0,j},b_{1,j},\ldots,b_{p-2,j}$ for $j=1,2,\ldots,k$ from column 1 to $k$, and compute the information bits $b_{0,0},b_{1,0},\ldots,b_{p-2,0}$. Then, we can compute $b_{i,k}$ from the information bits $b_{i,0},b_{i,1},\ldots,b_{i,k-1}$, and compute $a_{i,k}$ by $a_{i,k}=b_{i,k}+(a_{i,k}+b_{i,k})$ for $i=0,1,\ldots,p-2$. Together with $(k-1)(p-1)$ bits $a_{0,j},a_{1,j},\ldots,a_{p-2,j}$ with $j=1,2,\ldots,k$ from column 1 to $k$, we can compute $p-1$ information bits $a_{0,0},a_{1,0},\ldots,a_{p-2,0}$. The decoding method from any $k-1$ information columns plus any one parity column is similar.

Each parity column of the transformed EVENODD codes has optimal repair access. We can repair column $k$ by downloading $\mathbf{b_j}$ from column $j$ for $j=0,1,\ldots,k-1$ and $\mathbf{a_k}\oplus\mathbf{\bar{b}_k}\oplus\mathbf{b_k^*}$ from column $k+1$. Specifically, we can compute $\mathbf{b_k},\mathbf{b_{k+1}}$ from $\mathbf{b_0},\mathbf{b_1},\ldots,\mathbf{b_{k-1}}$, and $\mathbf{a_k}\oplus\mathbf{b_k}$ by $(\mathbf{a_k}\oplus\mathbf{\bar{b}_k}\oplus\mathbf{b_k^*})\oplus\mathbf{b_k^*}\oplus\mathbf{\bar{b}_k}\oplus\mathbf{b_k}$. Similarly, we can repair column $k+1$ by downloading $\mathbf{a_j}$ from column $j$ for $j=0,1,\ldots,k-1$ and $\mathbf{a_k}\oplus\mathbf{b_k}$ from column $k$. In the next theorem, we show that the efficient repair property of any single information column of EVENODD codes is preserved in the transformed EVENODD codes.

\begin{theorem}
In the $(k+2,k)$ EVENODD codes, suppose that $p-1$ is a multiple of four and we can download the bits $a_{i,j}$ for all $i\in S_j$ and $j=0,1,\ldots,f-1,f+1,\ldots,k+1$ to recover column $f$, where $0\leq f\leq k-1$, $S_j$ denotes the set of indices of the downloaded bits from column $j$ and $S_k=\{0,1,\ldots,(p-1)/2-1\}$.
Then, column $f$ of the transformed EVENODD codes can be recovered by downloading $a_{i,j},b_{i,j}$ for all $i\in S_j$ from column $j$ for $j=0,1,\ldots,f-1,f+1,\ldots,k-1$, $a_{i,k}+b_{i,k}$ for all $i\in S_k=\{0,1,\ldots,(p-1)/2-1\}$ and $b_{i,k+1}$ for all $i\in S_{k+1}$ from column $k$, $a_{i,k}+b_{i,k}+b_{i+1,k}$ and $a_{i+1,k}+b_{i,k}$ for all $i\in \{0,2,\ldots,(p-1)/2-2\}$ and $a_{i,k+1}$ for all $i\in S_{k+1}$ from column $k+1$.
\label{thm:evenodd-rep1}
\end{theorem}
\begin{proof}
Consider the repair of column $f$ for the transformed EVENODD codes. We have received the following bits
\begin{small}
\begin{align*}
\begin{bmatrix}
\text{Column } 0  & a_{i,0} \text{ } \forall i\in S_0 \text{ and } b_{i,0} \text{ } \forall i\in S_0 \\ \hline
\cdots & \cdots  \\ \hline
\text{Column } k-1 & a_{i,k-1} \text{ } \forall i\in S_{k-1}  \text{ and } b_{i,k-1} \text{ } \forall i\in S_{k-1} \\ \hline
\text{Column } k &  a_{i,k}+b_{i,k} \text{ } \forall i\in S_{k} \text{ and } b_{i,k+1} \text{ } \forall i\in S_{k+1} \\ \hline
\text{Column }  & a_{i,k}+b_{i,k}+b_{i+1,k},a_{i+1,k}+b_{i,k} \text{ }   \\
 k+1  &  \forall i\in \{0,2,\ldots,\frac{p-5}{2}\} \text{ and }  a_{i,k+1} \text{ } \forall i\in S_{k+1} \\
\end{bmatrix}.
\end{align*}
\end{small}
We can calculate $b_{i,k}$ by $(a_{i-1,k}+b_{i-1,k})+(a_{i-1,k}+b_{i-1,k}+b_{i,k})$ for $i=1,3,\ldots,(p-1)/2-1$, and $a_{i,k}$ by $b_{i,k}+(a_{i,k}+b_{i,k})$ for $i=1,3,\ldots,(p-1)/2-1$. Then, we can compute $b_{i,k}$ by $a_{i+1,k}+(a_{i+1,k}+b_{i,k})$ for $i=0,2,\ldots,(p-5)/2$ and $a_{i,k}$ by $b_{i,k}+(a_{i,k}+b_{i,k})$ for $i=0,2,\ldots,(p-5)/2$.
We thus obtain $a_{i,k}$ and $b_{i,k}$ for all $i\in S_k=\{0,1,\ldots,(p-1)/2-1\}$. Recall that we can recover $a_{0,f},a_{1,f},\ldots,a_{p-2,f}$ by downloading the bits $a_{i,j}$ for all $i\in S_j$ and $j=0,1,\ldots,f-1,f+1,\ldots,k+1$.
Therefore, we obtain the bits $a_{i,j},b_{i,j}$ for all $i\in S_j$ and $j=0,1,\ldots,f-,f+1,\ldots,k+1$,
and the bits $a_{0,f},a_{1,f},\ldots,a_{p-2,f}$ and $b_{0,f},b_{1,f},\ldots,b_{p-2,f}$ in column $f$ of the transformed EVENODD codes can be recovered.
\end{proof}
By Theorem \ref{thm:evenodd-rep1}, the efficient repair property of any information column of EVENODD codes with $r=2$ is preserved after the transformation, if $p-1$ is a multiple of four. When $r\geq 3$, the repair method of information column of EVENODD codes is different from that of EVENODD codes with $r=2$. We need to design new transformation carefully to preserve the efficient repair property of information column and that will be our future work.


\begin{table*}
\caption{The transformed EVENODD code with
$k=3$, $r=2$ and $p=5$.}
\vspace{-8pt}
\begin{center}
\begin{tabular}{|c|c|c|c|c|}
\hline
Column 0 & Column 1  & Column 2  & Column 3 & Column 4 \\
\hline
$a_{0,0},b_{0,0}$& $a_{0,1},b_{0,1}$ & $a_{0,2},b_{0,2}$& $a_{0,3}+b_{0,3},b_{0,4}$ & $a_{0,4},a_{0,3}+b_{0,3}+b_{1,3}$ \\
\hline
$a_{1,0},b_{1,0}$& $a_{1,1},b_{1,1}$ & $a_{1,2},b_{1,2}$& $a_{1,3}+b_{1,3},b_{1,4}$ & $a_{1,4},a_{1,3}+b_{0,3}$ \\
\hline
$a_{2,0},b_{2,0}$& $a_{2,1},b_{2,1}$ & $a_{2,2},b_{2,2}$& $a_{2,3}+b_{2,3},b_{2,4}$ & $a_{2,4},a_{2,3}+b_{2,3}+b_{3,3}$ \\
\hline
$a_{3,0},b_{3,0}$& $a_{3,1},b_{3,1}$ & $a_{3,2},b_{3,2}$& $a_{3,3}+b_{3,3},b_{3,4}$ & $a_{3,4},a_{3,3}+b_{2,3}$ \\
\hline
\end{tabular}
\end{center}
\label{table:evenodd2}
\end{table*}

Table \ref{table:evenodd2} shows an example of the transformed code with $k=3$, $r=2$ and $p=5$.
When $f=1$, we have $S_0=\{0,1,3\}$, $S_2=\{0,1,2\}$, $S_3=\{0,1\}$ and $S_4=\{1,3\}$ according to the repair method of the EVENODD code in Table \ref{table:evenodd1}. According to Theorem~\ref{thm:evenodd-rep1}, we can recover column 1 of the transformed EVENODD code by downloading the following 20 bits.
\begin{small}
\begin{align*}
&a_{0,0},a_{1,0},a_{3,0},a_{0,2},a_{1,2},a_{2,2},
b_{0,0},b_{1,0},b_{3,0},b_{0,2},b_{1,2},b_{2,2},a_{3,4},\\
&a_{0,3}+b_{0,3},a_{1,3}+b_{1,3},b_{1,4},b_{3,4},
a_{0,3}+b_{0,3}+b_{1,3},a_{1,3}+b_{0,3},a_{1,4}.
\end{align*}
\end{small}
Specifically, we can repair the bits $a_{0,1},a_{1,1}$ and $b_{0,1},b_{1,1}$ by
\begin{align*}
a_{0,1}=& a_{0,0}+a_{0,2}+ (a_{1,3}+b_{1,3})+(a_{1,3}+b_{0,3})+\\
&(a_{0,3}+b_{0,3}+b_{1,3}),\\
a_{1,1}=&a_{1,0}+a_{1,2}+(a_{0,3}+b_{0,3})+(a_{1,3}+b_{1,3})+\\
&(a_{0,3}+b_{0,3}+b_{1,3}),\\
b_{0,1}=&b_{0,0}+b_{0,2}+(a_{0,3}+b_{0,3})+(a_{1,3}+b_{1,3})+\\
&(a_{0,3}+b_{0,3}+b_{1,3})+(a_{1,3}+b_{0,3}),\\
b_{1,1}=& b_{1,0}+b_{1,2}+(a_{0,3}+b_{0,3})+(a_{0,3}+b_{0,3}+b_{1,3}),
\end{align*}
and repair $a_{2,1},a_{3,1},b_{2,1},b_{3,1}$ by
\begin{align*}
a_{3,1}=& a_{1,0}+a_{0,1}+a_{2,2}+a_{1,4},\\
a_{2,1}=& a_{3,0}+a_{1,2}+a_{3,1}+a_{2,2}+a_{3,4},\\
b_{3,1}=& b_{1,0}+b_{0,1}+b_{2,2}+b_{1,4},\\
b_{2,1}=& b_{3,0}+b_{1,2}+b_{3,1}+b_{2,2}+b_{3,4}.
\end{align*}
Therefore, we can recover column 1 by downloading 20 bits and the efficient repair property of column 1 is preserved in our transformation. We can also show that the efficient repair property of any other information column is preserved similarly.

We can also show that any one parity column of the transformed code is optimal. We can repair column 3 by downloading 16 bits
\begin{align*}
&b_{i,j} \text{ for } i=0,1,2,3 \text{ and } j=0,1,2, \text{ and } \\ &a_{0,3}+b_{0,3}+b_{1,3},a_{1,3}+b_{0,3},a_{2,3}+b_{2,3}+b_{3,3},a_{3,3}+b_{2,3}.
\end{align*}
Specifically, we can compute $b_{i,3},b_{i,4}$ from $b_{i,0},b_{i,1},b_{i,2}$ for $i=0,1,2,3$, as EVENODD is MDS code. Then, we can compute the other four bits in column 3 by
\begin{align*}
a_{0,3}+b_{0,3}=&(a_{0,3}+b_{0,3}+b_{1,3})+b_{1,3},\\
a_{1,3}+b_{1,3}=&(a_{1,3}+b_{0,3})+b_{0,3}+b_{1,3},\\
a_{2,3}+b_{2,3}=&(a_{2,3}+b_{2,3}+b_{3,3})+b_{3,3},\\
a_{3,3}+b_{3,3}=&(a_{3,3}+b_{2,3})+b_{2,3}+b_{3,3}.
\end{align*}
Therefore, the repair access of the first parity column is optimal. Similarly, we can repair column 4 by downloading 16 bits
\begin{align*}
&a_{i,j} \text{ for } i=0,1,2,3 \text{ and } j=0,1,2, \text{ and } \\ &b_{0,3}+a_{0,3},b_{1,3}+a_{1,3},b_{2,3}+a_{2,3},b_{3,3}+a_{3,3},
\end{align*}
and the repair access of column 4 is optimal.

In order to obtain binary MDS array codes with low sub-packetization that have efficient repair for any column, we show in this section how to apply the transformation for the array codes in \cite{hou2018c} and EVENODD codes. Note that the transformation given in this section can be viewed as a variant of the transformation in Section \ref{sec:trans}. We can also apply the transformation given in this section multiple times for EVENODD codes to obtain the multi-layer transformed EVENODD codes that have optimal repair for any column, as the construction in Section \ref{sec:cons}. The difference between two transformations is that, the efficient repair property of any information column of codes in \cite{hou2018c} is maintained with the transformation given in this section, while not for the transformation in Section \ref{sec:trans}.
The relationship of sub-packetization and repair bandwidth of binary MDS array codes is one of our future work.

\section{Discussion and Conclusion}
\label{sec:discussions}
In this paper, we propose a generic transformation for EVENODD codes that can enable optimal repair access for the chosen $d-k+1$ columns. Based on the proposed EVENODD transformation, we present the multi-layer transformed $\mathsf{EVENODD}_{\lceil
\frac{k}{d-k+1}\rceil+\lceil \frac{r}{d-k+1}\rceil}$ that have optimal repair access for all $k+r$ columns. In $\mathsf{EVENODD}_{\lceil
\frac{k}{d-k+1}\rceil+\lceil \frac{r}{d-k+1}\rceil}$, the $d$ helper columns can be selected from $k+1$ and $k+r-1$, and some of the $d$ helper columns should be specifically selected.
Moreover, we show that the proposed transformation can also be employed in other existing binary MDS array codes, such as codes in \cite{blomer1999,feng2005,hou2018a,schindelhauer2013,hou2017,hou2018c,hou2018b},
that can enable optimal repair access.
How to combine the existing binary MDS array codes with asymptotically optimal repair access by our transformation to obtain the transformed binary MDS array codes with asymptotically optimal repair access for all columns and lower sub-packetization is an interesting and practical future work. The implementation of the proposed transformed binary MDS array codes in practical storage systems is another one of our future works.

\ifCLASSOPTIONcaptionsoff
  \newpage
\fi

\bibliographystyle{IEEEtran}

\end{document}